\documentclass{article}
\let\standardaddcontentsline\addcontentsline
\usepackage{iclr2027_conference,times}
\let\addcontentsline\standardaddcontentsline
\usepackage{amsmath}
\usepackage{amsfonts}
\usepackage{mathrsfs}
\usepackage{amssymb,bbm, color}
\usepackage{enumerate}
\usepackage[hidelinks,colorlinks=true,linkcolor=blue,citecolor=blue]{hyperref}
\usepackage{graphicx}
\numberwithin{equation}{section}
 \usepackage{natbib}
\providecommand{\U}[1]{\protect\rule{.1in}{.1in}}

\newtheorem{theorem}{Theorem}[section]

\newtheorem{algorithm}[theorem]{Algorithm}

\newtheorem{definition}[theorem]{Definition}
\newtheorem{example}[theorem]{Example}

\newtheorem{lemma}[theorem]{Lemma}

\newtheorem{proposition}[theorem]{Proposition}
\newtheorem{remark}[theorem]{Remark}

\newenvironment{proof}[1][Proof]{\noindent \textbf{#1.} }{\  \rule{0.5em}{0.5em}}
\newenvironment{restatedresult}[1]
  {\begin{trivlist}\item[\hskip\labelsep\bfseries #1 (restated).]\itshape}
  {\end{trivlist}}

\newcommand{\RR}{\mathbb{R}}

\newcommand{\dd}{{\rm d}}
\def\conv{{\operatorname{conv}}}
\newcommand{\cD}{{\mathcal{D}}}
\newcommand{\partialc}{\partial^c} 
\newcommand{\NN}{\mathbb{N}}

\DeclareMathOperator*{\prox}{prox}

\DeclareMathOperator*{\TV}{TV}
\def\1{{\mathbbm 1}}

\allowdisplaybreaks

\title{MASLA: Metropolis-Adjusted Subdifferential Langevin Sampling for Nonsmooth Nonconvex Potentials}
\author{Ning Ning\\Department of Statistics\\Texas A\&M University\\\texttt{patning@tamu.edu}}
\iclrfinalcopy
\begin{document}
\maketitle
\lhead{} % Omit the conference publication header for the arXiv version.
\begin{abstract}
Sampling from distributions with nonsmooth nonconvex potentials challenges gradient-based Langevin methods and can complicate proximal approaches. We introduce Metropolis-adjusted subdifferential Langevin algorithm (MASLA), which combines a fixed measurable selection from a conservative field with an exact Metropolis--Hastings correction. The method accommodates locally Lipschitz potentials without requiring smoothing or a proximal mapping. We establish well-posedness, reversibility, and exact invariance of the target distribution, together with positive Harris recurrence. Our central theoretical result rigorously proves geometric ergodicity for positive-definite quadratic potentials perturbed by  Lipschitz functions admitting bounded conservative fields. The result provides an explicit admissible range of step sizes and holds from every initial point, including points of nondifferentiability. Numerical experiments verify the predicted geometric convergence behavior on a convex composite target and demonstrate competitive performance against proximal and subgradient samplers. Further experiments on a nonsmooth nonconvex target show that the Metropolis correction substantially reduces sampling error relative to unadjusted subdifferential Langevin dynamics. These results establish MASLA as a theoretically grounded sampling method for a class of nonsmooth nonconvex distributions.
\end{abstract}

\section{Introduction}
\label{sec:main_intro}
Sampling from a density $\pi(x)\propto e^{-U(x)}$ is central to Bayesian computation, where posterior expectations and uncertainty quantification require more than a single optimizer of the potential. Langevin methods use local derivative information to guide exploration of the target distribution. When $U$ is differentiable, the Metropolis-adjusted Langevin algorithm (MALA) uses an Euler--Maruyama step of Langevin diffusion as a proposal and applies a Metropolis--Hastings (MH) correction. This combination exploits the geometry of the target while preserving its invariant distribution. The corresponding unadjusted Langevin algorithm (ULA) can exhibit discretization bias and can even lose stability for smooth potentials \citep{Roberts1996Exponential}. Section~\ref{sec:motivation} illustrates this distinction and explains why we retain the MH correction when extending Langevin sampling to nonsmooth problems.

\paragraph{Nonsmooth and nonconvex sampling.}
Nonsmooth potentials arise naturally from sparsity-inducing priors, composite regularizers, and models involving nonsmooth activations. For example, the Bayesian analogue of an $\ell_1$ penalty is a Laplace prior, whose negative log-density is nondifferentiable at zero \citep{park2008bayesian,seeger2008bayesian}. More general model structures and penalties can also introduce nonconvexity. These features complicate classical gradient-based sampling: a gradient may not exist at every state, and tools relying on convexity may no longer apply. A useful extension should accommodate generalized derivatives, specify the update at exceptional points, and retain a rigorous connection between the implemented Markov chain and the intended target.

\paragraph{Related approaches.}
Proximal and Moreau--Yosida methods address nonsmooth sampling through implicit updates or smoothing \citep{pereyra2016proximal,bernton2018langevin,durmus2018efficient,durmus2019analysis}; proximal MH methods also support composite posterior applications \citep{abry2023proximal}. Subgradient methods avoid some proximal computations for convex composite potentials \citep{habring2024subgradient}, and \citet{lehec2021langevin} gives nonasymptotic bounds for Lipschitz convex potentials without Lipschitz gradients. \citet{lau2023nonlogconcave} compares Langevin algorithms for nonsmooth non-log-concave targets. Inexact proximal algorithms quantify errors from approximate subproblem solves \citep{ehrhardt2024inexact,ghaderi2026inexact}; absence of a closed-form prox therefore does not make proximal MCMC inapplicable. Other approaches change the driving noise \citep{oechsler2024levy} or use a Hadamard parameterization of an $\ell_1$ prior \citep{cheltsov2026hadamard}. The latter establishes geometric ergodicity for continuous dynamics and convergence of a discretization as its step size vanishes. \citet{crucinio2025scaling} studies dimension scaling for Moreau--Yosida MH proposals; such scaling results are distinct from fixed-dimensional geometric ergodicity. \citet{srinivasan2024mirror} gives mixing bounds for MH-corrected mirror Langevin sampling on compact convex domains, a different geometry from our unconstrained nonsmooth potentials. 

\paragraph{Our approach and contributions.}
To address nonsmooth and nonconvex targets, we introduce the Metropolis-adjusted subdifferential Langevin algorithm (MASLA), built on the framework of conservative fields. Conservative fields provide a notion of generalized derivatives that accommodates a broad class of nonsmooth nonconvex functions while retaining a chain rule \citep{bolte2021conservative}. MASLA uses a fixed measurable selection from a conservative field to construct proposals and an exact MH correction to preserve the target distribution, without requiring a proximal mapping or smoothing parameter. We establish its discrete-time validity directly, independently of existence results for continuous-time stochastic differential inclusions \citep{difonzo2022stochastic}.
 Our contributions are threefold:
\begin{itemize}
\item \textbf{Exact sampling with conservative-field proposals.} We specify the MASLA kernel and establish exact invariance, reversibility, and positive Harris recurrence for every positive step size. We also characterize when the chain law is independent of subgradient choices.
\item \textbf{Rigorous geometric ergodicity guarantees.} We prove geometric ergodicity for a positive-definite quadratic potential plus a  Lipschitz perturbation with a bounded conservative field, with an explicit step-size range and convergence from every initial point, including points of nondifferentiability.
\item \textbf{Numerical verification and sampling accuracy.} Experiments verify the predicted geometric convergence behavior on a convex composite target, compare MASLA with proximal and subgradient samplers, and demonstrate the accuracy gained from MH correction on a nonsmooth nonconvex target.
\end{itemize}
% Section~\ref{sec:main_algorithm} defines the method, Section~\ref{sec:main_theory} presents the guarantees, and Section~\ref{sec:main_experiments} reports the experiments. Full proofs and supporting details appear in the appendices.

\section{Motivation: why retain the Metropolis correction?}
\label{sec:motivation}
%\subsection{Discretization can destroy stability}
For a smooth potential, Langevin diffusion and its Euler discretization take the forms
\begin{equation}
 dX_t=-\nabla U(X_t)\,dt+\sqrt{2}\,dB_t,
 \qquad Y_{n+1}=Y_n-\gamma\nabla U(Y_n)+\sqrt{2\gamma}\,Z_{n+1}.
 \label{eq:motivation_langevin}
\end{equation}
Here $(X_t)_{t\geq0}$ is the continuous-time process driven by a standard $d$-dimensional Brownian motion $(B_t)_{t\geq0}$, and $(Y_n)_{n\geq0}$ is the ULA chain with fixed step size $\gamma>0$. The noise vectors $(Z_n)_{n\geq1}$ are independent and identically distributed (i.i.d.) as $N(0,I_d)$, independently of $Y_0$, where $I_d$ is the $d\times d$-dimensional identity matrix. 
Even when the diffusion is well behaved, the discrete chain need not inherit its stability or invariant law \citep{Roberts1996Exponential}. We use a scaled version of the Gaussian example discussed in Sections~1.4.1 and~3 of that work (their step size $h$ equals $2\gamma$ in our convention). For $\pi=N(0,\sigma^2)$, the potential is $U(x)=x^2/(2\sigma^2)$ and ULA becomes
\begin{equation}
Y_{n+1}=\left(1-\frac{\gamma}{\sigma^2}\right)Y_n+\sqrt{2\gamma}\,Z_{n+1}.
\label{eq:motivation_gaussian_ula}
\end{equation}
This autoregression is stable for $0<\gamma/\sigma^2<2$, with stationary variance $\sigma^2/(1-\gamma/(2\sigma^2))$, and is explosive when $\gamma/\sigma^2>2$. We set $\sigma=0.02$ and $\gamma=0.001$, giving coefficient $-1.5$. Thus, although the absolute step is small, it exceeds the stability limit $2\sigma^2=0.0008$. Reducing $\gamma$ below this limit restores ULA stability for this fixed target, albeit with discretization bias.

\begin{figure}[tb]
\centering
\includegraphics[width=\linewidth]{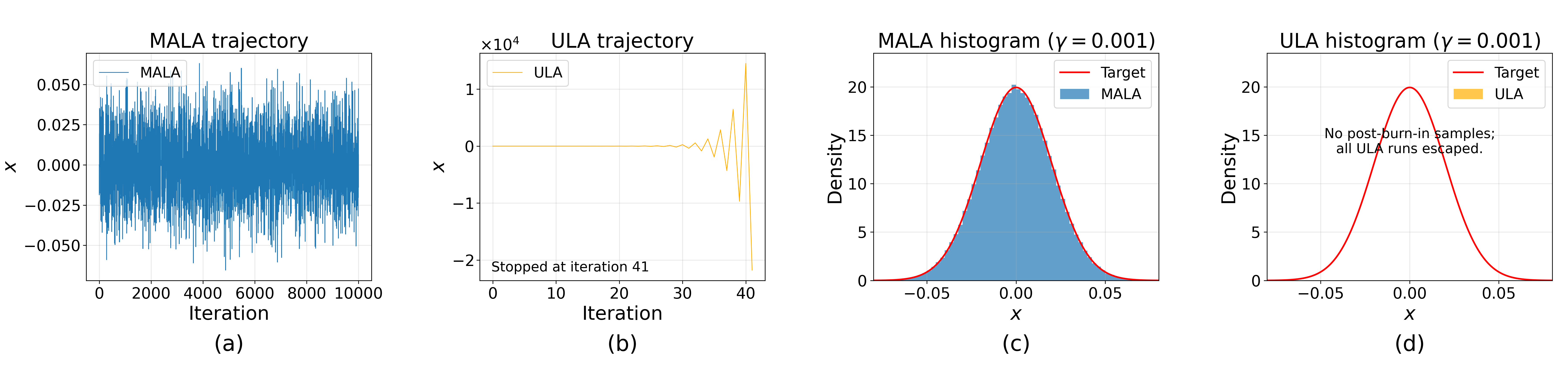}
\caption{MALA and ULA for $\pi=N(0,0.02^2)$ with $\gamma=0.001$, both initialized at zero. (a)--(b) The first replicate's trajectories: MALA is shown for 10,000 updates, and ULA stops at iteration 41 upon crossing $|x|/\sigma\geq10^6$; the axis scales differ. (c)--(d) Histograms pool the retained observations from 16 replicates, each with a nominal length of 100,000 updates and 20,000 burn-in updates; the target density is red. All ULA replicates stop before burn-in ends, so panel (d) has no histogram bars.}
\label{fig:mala_vs_ula}
\end{figure}

Figure~\ref{fig:mala_vs_ula} shows that MALA remains in the central target region and its histogram closely matches the Gaussian density, whereas all 16 ULA runs cross the numerical escape threshold within 31--41 updates. MALA has an average acceptance rate of $39.6\%$ and a pooled sample variance of $1.0019$ for $X/\sigma$, whose target variance is one. Its MH correction preserves $\pi$ exactly by accounting for the target density and the forward--reverse proposal ratio. The experiment illustrates the benefit of this correction when the Euler proposal is unstable; finite-run agreement does not by itself establish a geometric convergence rate. Appendix~\ref{sec:gaussian_motivation_details} gives the simulation details and the analytical basis for ULA transience.

%\subsection{From gradients to conservative fields}
%For nonsmooth potentials, replacing the gradient by a generalized derivative addresses how to form a proposal, but leaves the discretization issue in place. MASLA therefore combines two ingredients: a fixed measurable conservative-field selection supplies the drift, and an exact MH step corrects the proposal. This preserves the original target without proximal regularization. The remaining theoretical question is whether the accepted chain converges geometrically. Our analysis answers it for a specified nonsmooth nonconvex class by controlling both the proposal and the rejection probability. The experiments in Section~\ref{sec:main_experiments} then examine the resulting convergence behavior and the effect of the correction in nonsmooth settings.

\section{Conservative fields and the MASLA algorithm}
\label{sec:main_algorithm}
Throughout, the state space is $\mathbb R^d$, $\lambda$ is Lebesgue measure, and $U:\mathbb R^d\to\mathbb R$ is finite-valued and Borel measurable with
\begin{equation}
	0<Z_U:=\int e^{-U(x)}\,\lambda(dx)<\infty,
	\qquad \pi(x)=Z_U^{-1}e^{-U(x)}.
	\label{eq:standing_target}
\end{equation}
For the conservative-field formulation, $U$ is locally Lipschitz and admits a conservative field $D_U:\mathbb R^d\rightrightarrows\mathbb R^d$. The double arrow assigns a set of vectors $D_U(x)\subseteq\mathbb R^d$ to each $x$, whereas $\to$ denotes a single-valued map. We fix one finite-valued Borel measurable selection $\beta:\mathbb R^d\to\mathbb R^d$ with $\beta(x)\in D_U(x)$ at every point. A locally bounded conservative field has closed graph, nonempty compact values, and the chain rule along absolutely continuous curves (Appendix~\ref{sec:Preliminaries}). Convex and locally Lipschitz semialgebraic functions are examples of potentials admitting such fields.

MASLA is formulated based on a generalized, set-valued version of the Langevin diffusion:
\begin{equation}
	\label{eqn:SLD}
	d X_t \in -D_U(X_t)dt + \sqrt{2}\, d B_t.
\end{equation}
Unlike the classical MALA algorithm, where proposals are generated using a single gradient as in \eqref{eq:motivation_langevin}, MASLA draws a subgradient from the set-valued field \( D_U(x) \) at each iteration:
\begin{equation}
	\label{eqn:SELD}
	\breve{Y}_{n+1} = X_n - \gamma \beta(X_n) + \sqrt{2 \gamma}\, Z_{n+1}, \qquad \beta(X_n) \in D_U(X_n),
\end{equation}
which is then followed by the standard MH accept-reject step. 
Let $\cD_f$ denote the set of points $x\in
\RR^d$ where $f$ is differentiable.
A special case of \eqref{eqn:SELD} is the following:
\begin{equation}
	\label{eqn:proposal_0}
	\breve{Y}_{n+1} = X_n - \gamma \beta(X_n)+ (2 \gamma)^{1/2} Z_{n+1}, 
\end{equation}	
where 
\begin{equation}
	\begin{cases}
		\beta(x) = \nabla U(x)\quad&x\in \cD_U,  \\
		\beta(x) \in D_U(x)& \text{otherwise.}
	\end{cases}
\end{equation}
Remark \ref{remark:special_case} and Example \ref{example:special_case} explain why \eqref{eqn:proposal_0} is special case instead of an equivalent form. 
\begin{remark}
	\label{remark:special_case}
	The inclusion $\beta(x)\in D_U(x)$ holds for all $x\in \mathbb{R}^d$, because $\nabla U(x)$ is an element of $D_U(x)$ when the former exists. However, conversely, a function $\beta(x)\in D_U(x)$ does not necessarily satisfy $\beta(x)= \nabla U(x)$ if $x\in \cD_U$. 
	The example below shows that differentiability alone does not imply regularity.
\end{remark}

\begin{example}
	\label{example:special_case}
	Let $g:\mathbb{R}\to\mathbb{R}$ be defined by
	\[
	g(x) :=
	\begin{cases}
		x^{2}\sin(1/x), \quad& x\neq 0,\\
		0, & x=0.
	\end{cases}
	\]
	Then $g$ is locally Lipschitz on $\mathbb{R}$ and differentiable at $x=0$ with
	\[
	\nabla g(0)=0.
	\]
	However, its Clarke subdifferential at $0$ satisfies
	\[
	\partial g(0)=[-1,1].
	\]
\end{example}
The proofs to Example \ref{example:special_case} and Proposition \ref{thm:stationary} are provided in Appendix \ref{sec:Wellposedness}. 
\begin{proposition}[Well-posedness and exact invariance]
\label{thm:stationary}
Let $U:\mathbb R^d\to\mathbb R$ be finite-valued and Borel measurable and satisfy 
\begin{equation}
	0<Z_U:=\int e^{-U(x)}\,\lambda(dx)<\infty,
	\qquad \pi(x)=Z_U^{-1}e^{-U(x)}.
	\label{eq:standing_target}
\end{equation}
Fix a finite-valued Borel map $\beta:\mathbb R^d\to\mathbb R^d$ (in MASLA, a selection of $D_U$). For every $\gamma>0$, define
\begin{equation}
	q_\gamma(x,y)=(4\pi\gamma)^{-d/2}
	\exp\left(-\frac{\|y-x+\gamma\beta(x)\|^2}{4\gamma}\right),
	\label{eq:masla_proposal_density}
\end{equation}
\begin{align}
	\alpha_\gamma(x,y)&=1\wedge\frac{\pi(y)q_\gamma(y,x)}{\pi(x)q_\gamma(x,y)},
	& r_\gamma(x)&=1-\int q_\gamma(x,y)\alpha_\gamma(x,y)\,\lambda(dy).
	\label{eq:masla_acceptance}
\end{align}
Then the Markov transition kernel
\begin{equation}
	R_\gamma(x,A)=\int_A q_\gamma(x,y)\alpha_\gamma(x,y)\,\lambda(dy)+r_\gamma(x)\delta_x(A),
	\quad A\in\mathcal B(\mathbb R^d),
	\label{eq:masla_kernel}
\end{equation}
is well-defined from every initial law, reversible with respect to $\pi$, and admits $\pi$ as an invariant probability distribution.
\end{proposition}
Explicit MASLA pseudocode is given in Algorithm~\ref{alg:masla_explicit} of Appendix~\ref{app:implementation}; Section~\ref{sec:autodiff_selection} discusses choosing $\beta$ by automatic differentiation.

\section{Convergence guarantees}
\label{sec:main_theory}
\subsection{Selection independence and recurrence}
A conservative field can contain more than the classical gradient at an exceptional differentiability point. Thus, almost-everywhere agreement of fields does not justify arbitrary pointwise substitutions. The following result makes the required choice explicit.

\begin{proposition}[Selection independence]
\label{prop:selection_independence} 
Suppose $U$ is locally Lipschitz and satisfies \eqref{eq:standing_target}, and $D_U$ is a conservative field for $U$. 
Let $N_U$ be its nondifferentiability set. Consider finite-valued Borel selections $\beta_1,\beta_2$ of $D_U$ with
\[
\beta_1(x)=\beta_2(x)=\nabla U(x),\qquad x\notin N_U,
\]
with possibly different values on $N_U$. For $i\in\{1,2\}$ and
$\gamma>0$, let $R_{\gamma,i}$ be the MASLA kernel constructed from
$\beta_i$ as in \eqref{eq:masla_proposal_density}--\eqref{eq:masla_kernel}.
Then, for every $\gamma>0$,
\begin{equation}
\label{eq:selection_kernel_equality}
R_{\gamma,1}(x,\cdot)=R_{\gamma,2}(x,\cdot),
\qquad x\notin N_U.
\end{equation}
In particular, the kernels agree for $\lambda$-almost every $x$ and
for $\pi$-almost every $x$.

If both chains have the same initial law $\mu$ with $\mu(N_U)=0$,
then their laws on the sequence space $(\mathbb{R}^d)^{\{0,1,\ldots\}}$
coincide, and each chain satisfies
\[
\mathbb{P}_\mu\bigl(X_n\notin N_U\text{ for every }n\geq0\bigr)=1.
\]
This applies, in particular, to every absolutely continuous initial law
and to every deterministic initial point outside $N_U$. Moreover, an
absolutely continuous initial law remains absolutely continuous at every
iteration under either kernel.
\end{proposition}

The restriction on the initial law matters: from a nonsmooth initial point, different selections can change both the proposal and the rejection probability. Exact invariance still holds for each fixed selection. The next result does not require the initial state to avoid the nonsmooth set.

\begin{proposition}[Irreducibility and Harris recurrence]
\label{prop:harris_recurrence}
Under conditions of Proposition~\ref{thm:stationary}, including normalizability and a finite-valued Borel map $\beta$, the MASLA kernel $R_\gamma$ is
$\lambda$-irreducible and $\pi$-irreducible for every $\gamma>0$. More precisely, $R_\gamma(x,A)>0$ for every
$x\in\mathbb{R}^d$ and every Borel set $A$ with $\lambda(A)>0$.
Furthermore, the chain is positive Harris recurrent: for every such
set $A$ and every initial point $x$,
\[
\mathbb{P}_x\bigl(X_n\in A\text{ for infinitely many }n\bigr)=1,
\]
and it admits an invariant probability distribution. Its unique
invariant probability distribution is $\pi$.
\end{proposition}

Proofs to Propositions \ref{prop:selection_independence}  and \ref{prop:harris_recurrence} are provided in Appendix \ref{sec:Convergence}. When $U$ and $\beta$ are bounded on compact sets, every compact set of positive Lebesgue measure is a one-step small set and the chain is strongly aperiodic (Proposition~\ref{prop:Aperiodicity}). These properties support the geometric drift argument below.

\subsection{Geometric ergodicity for a nonconvex class}
Exact invariance and recurrence alone do not establish a geometric convergence rate. A drift inequality must hold for the full accepted chain, including its rejection term. Theorem~\ref{thm:geometric_ergodicity} in the appendix states the general drift and minorization criterion. Here we give conditions directly on the potential and field that imply that criterion.

\begin{theorem}[MASLA geometric ergodicity]
\label{thm:masla_geometric_ergodicity}
Let $H\in\mathbb{R}^{d\times d}$ be symmetric and positive definite.
Write $h_-$ and $h_+$ for its smallest and largest eigenvalues.
Suppose that
\[
U(x)=\tfrac12 x^\top Hx+g(x),
\qquad |g(x)-g(y)|\leq L\|x-y\|,
\qquad x,y\in\mathbb{R}^d,
\]
for some $L<\infty$. Assume that $g$ admits a conservative field
$D_g$ whose values are uniformly bounded: for some $B<\infty$,
\[
\sup_{x\in\mathbb{R}^d}\sup_{v\in D_g(x)}\|v\|\leq B.
\]
Choose any fixed Borel measurable selection $v(x)\in D_g(x)$ and set
\[
D_U(x):=Hx+D_g(x),\qquad \beta(x):=Hx+v(x).
\]
For every step size
\begin{equation}
\label{eq:masla_quadratic_stepsize}
0<\gamma<\frac{2}{h_+},
\end{equation}
the MASLA kernel using this selection and the exact MH acceptance
probability is $V$-uniformly geometrically ergodic with
$V(x)=1+\|x\|^2$. Specifically, $\pi(dx)\propto e^{-U(x)}\lambda(dx)$
is well-defined, and there exist $M_\gamma<\infty$ and
$\rho_\gamma\in(0,1)$ such that
\[
\|R_\gamma^n(x,\cdot)-\pi\|_V
\leq M_\gamma(1+\|x\|^2)\rho_\gamma^n,
\qquad x\in\mathbb{R}^d,\quad n\geq0.
\]
This includes initial points at which $U$ is not differentiable.
More precisely, writing $a_\gamma:=\|I_d-\gamma H\|_{\mathrm{op}}<1$,
the rejection probability and the drift satisfy
\begin{equation}
\label{eq:masla_quadratic_tail_bounds}
\lim_{\|x\|\to\infty}r_\gamma(x)=0,
\qquad
\limsup_{\|x\|\to\infty}\frac{R_\gamma V(x)}{V(x)}
\leq a_\gamma^2<1.
\end{equation}
\end{theorem}

Here $\|\xi\|_V=\sup_{|f|\leq V}|\int f\,d\xi|$, with the supremum over Borel functions. The constants may depend on $\gamma$; the theorem does not assert a rate uniform over all step sizes or dimensions.

\paragraph{Proof outline.}
The quadratic lower bound makes $e^{-U}$ integrable. Write the proposal as
$Y=(I_d-\gamma H)x-\gamma v(x)+\sqrt{2\gamma}Z$.
Expanding the log MH ratio gives
\[
\ell_\gamma(x,y)=\frac{\gamma}{4}\bigl(\|Hx\|^2-\|Hy\|^2\bigr)+E_\gamma(x,y),
\qquad |E_\gamma(x,y)|\leq K_\gamma(\|x\|+\|y\|)+\gamma B^2/4.
\]
On every bounded-noise event, the positive quadratic term dominates the linear error as $\|x\|\to\infty$. Hence $r_\gamma(x)\to0$. Retaining the rejection contribution gives
\[
R_\gamma V(x)\leq\mathbb E[V(Y)]+r_\gamma(x)V(x),
\qquad
\limsup_{\|x\|\to\infty}\frac{R_\gamma V(x)}{V(x)}\leq a_\gamma^2<1.
\]
A compact-set minorization completes the argument. Appendix~\ref{sec:Convergence} supplies every estimate.

\paragraph{A nonsmooth nonconvex example.}
For $\tau>0$, take $U(x)=\tfrac12x^\top Hx-\tau\|x\|_1$. The perturbation has Lipschitz constant $\tau\sqrt d$ and admits the bounded conservative field
\[
D_g(x)=\prod_{j=1}^d G_j(x_j),\qquad
G_j(t)=\begin{cases}\{-\tau\},&t>0,\\{}[-\tau,\tau],&t=0,\\\{\tau\},&t<0.\end{cases}
\]
The potential is nonsmooth on coordinate hyperplanes and nonconvex near the origin, yet the theorem applies to every Borel selection and every $0<\gamma<2/h_+$. Example~\ref{ex:masla_nonconvex_geometric} gives the verification. Thus the rate theorem covers a concrete nonconvex class, while the basic kernel construction applies more broadly.

\section{Numerical experiments}
\label{sec:main_experiments}
We evaluate MASLA on a convex benchmark against established samplers (Section~\ref{sec:main_convex_comparison}), a one-dimensional nonconvex target to assess MH correction (Section~\ref{sec:main_nonconvex_comparison}), and a high-dimensional quartic target to examine numerical stability (Section~\ref{sec:main_high_dimensional_quartic}).

\subsection{Comparison where competing methods apply}
\label{sec:main_convex_comparison}
\paragraph{Target and experimental protocol.}
	To ensure a fair comparison, we adopt the exact experimental setup in \cite{habring2024subgradient} on sampling from a two-dimensional density with the potential 
$$U(x) = F(x) + G(Kx),$$ 
where $F: \mathbb{R}^d \rightarrow [0, \infty)$ and $G: \mathbb{R}^{d'} \rightarrow [0, \infty)$ are both proper, convex, and lower semi-continuous, and $K: \mathbb{R}^d \rightarrow \mathbb{R}^{d'}$ is a linear operator. They considered function $F$ as either Lipschitz continuous with constant $L_F$ or differentiable with a Lipschitz continuous gradient with constant $L_{\nabla F}$, while $G$ as convex and Lipschitz continuous with constant $L_G$.  They used the total variation (TV) functional for $G(Kx)$, defined in the two-dimensional case as $G: \mathbb{R} \to \mathbb{R}$ and $K: \mathbb{R}^2 \to \mathbb{R}$ by
\[
G(p) = \lambda |p|, \qquad Kx = x_2 - x_1,
\]
where $\lambda > 0$ is a scaling parameter.  It is straightforward to verify that $G$ is Lipschitz continuous with Lipschitz constant $L_G = \lambda$, and its subdifferential is given by $\partial G(p)=\{\lambda\operatorname{sign}(p)\}$ for $p\ne0$ and $[-\lambda,\lambda]$ at zero.
The operator $K$ satisfies $\|K\|^2 \leq 2$. For $F: \mathbb{R}^2 \to \mathbb{R}$, they used the squared $L^2$ discrepancy, defined as
\[
F(x) = \frac{1}{2\sigma^2} \|x - y\|_2^2,
\]
for $y$ being a known vector,
which is $m$-strongly convex with $m = \sigma^{-2}$ and has a Lipschitz continuous gradient with constant $L_{\nabla F} = \sigma^{-2}$. Following \cite{habring2024subgradient}, we set $y = (-1, 1)$, $\sigma = 1$, and $\lambda = 5$. 

 For each algorithm, 10,000 independent chains start at $(-1,1)$; histograms across chains estimate the distribution at a fixed iterate. These comparisons use no burn-in removal and test step sizes $\tau\in\{10^{-3},10^{-4},10^{-5}\}$. The first-coordinate distribution is initially concentrated near $-1$ and approaches the target marginal after 49,500 iterations; the initial and final density plots are given in Appendix~\ref{app:experiments}. These are distributions across independent chains at a fixed iteration, rather than time-averaged samples from a single trajectory.

\paragraph{Wasserstein comparison.}
	\cite{habring2024subgradient} proposed two algorithms: the Proximal-subgradient Langevin algorithm (Prox-sub) and the Gradient-subgradient Langevin algorithm (Grad-sub).  In Figure 1 of therein, the authors investigated the convergence of Grad-sub and Prox-sub in Wasserstein-2 distance
$
W_2(\mu,\nu)=\left(\inf_{\zeta\in\Pi(\mu,\nu)}\int\|x-y\|^2\,d\zeta(x,y)\right)^{1/2}$,
where $\Pi(\mu,\nu)$ is the set of couplings of two probability measures $\mu,\nu$. Appendix~\ref{sec:numerical_distance_calculation} details the discretized reference distribution and optimal-transport calculation of the plotted squared Wasserstein errors.   Figure \ref{fig:Comparison_W2} shows that numerical convergence rates exceed theoretical expectations (explained in Appendix~\ref{sec:expected_wasserstein_rates}) for all $\tau$ values. Experimental results indicate negligible differences in convergence rates among Prox-sub, Grad-sub, and MASLA. All algorithms exhibit exponential ergodicity to the target density. 
%
% explains the expected-rate bounds for Grad-sub and Prox-sub, their numerical constants, and their interpretation as upper bounds on $W_2^2$.
\begin{figure*}[tb]
		\centering
		{\includegraphics[width=0.85\linewidth]{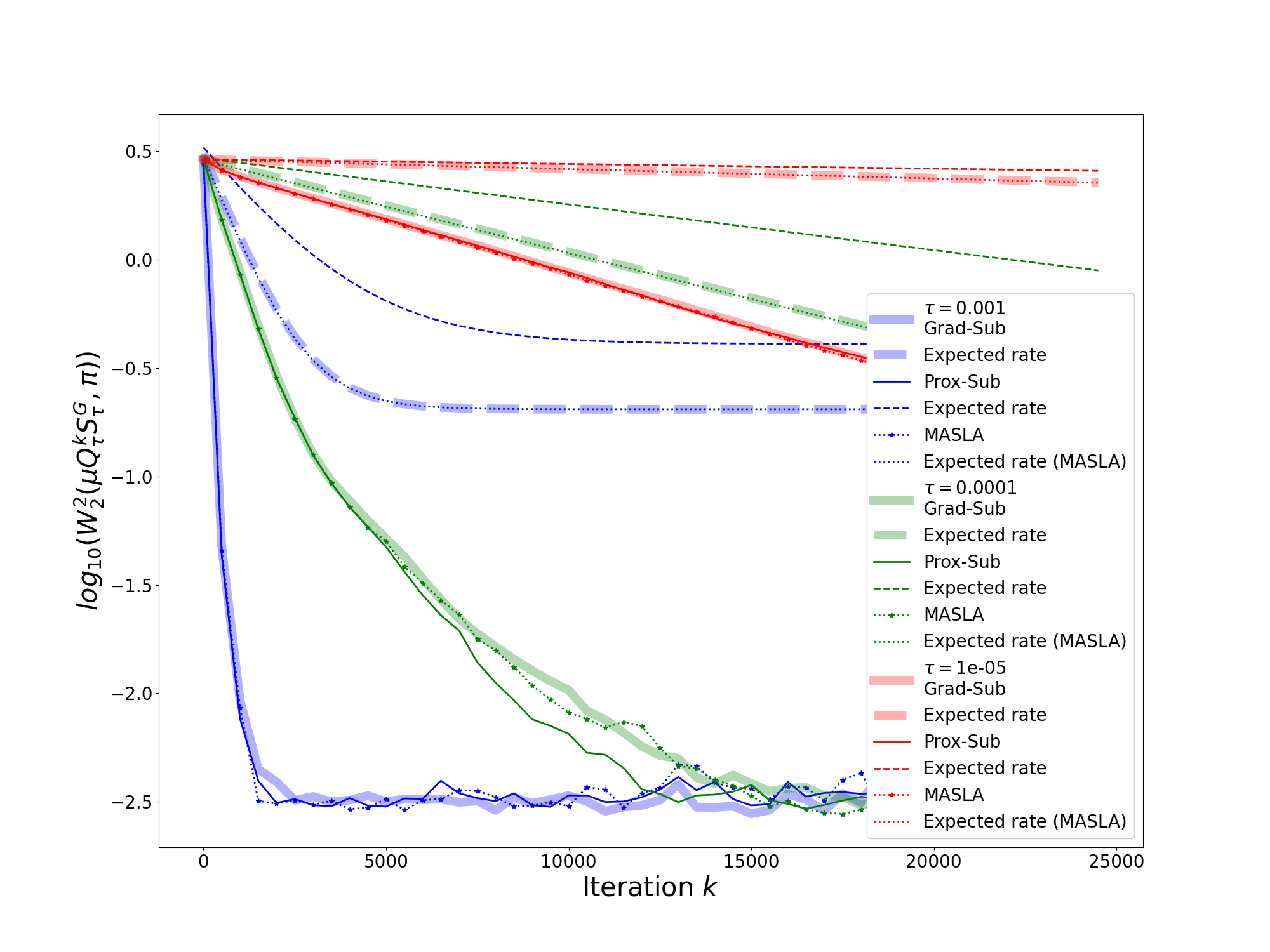}}
		\caption{Comparison of MASLA, Grad-sub \citep{habring2024subgradient}, and Prox-sub \citep{habring2024subgradient} for $\TV-L^2$ sampling in Wasserstein distance. Colors denote different $\tau$ values, dashed lines represent theoretical convergence rates, and solid lines show empirical results. Thick transparent lines indicate Grad-sub, thin solid lines indicate Prox-sub, and dotted lines indicate MASLA.}
		\label{fig:Comparison_W2}
	\end{figure*}

\paragraph{Total variation comparison.}

In Figure~3 of \cite{habring2024subgradient}, the authors compared Prox-sub with the Proximal Metropolis-Adjusted Langevin Algorithm (P-MALA) of \cite{pereyra2016proximal} and the Moreau--Yosida Unadjusted Langevin Algorithm (MYULA) of \cite{durmus2022proximal} in terms of total variation (TV) distance. The TV distance between two probability measures $\mu$ and $\nu$ is defined as
$
\mathrm{TV}(\mu, \nu) = \sup_{A \in \mathcal{F}} |\mu(A) - \nu(A)|$. Appendix~\ref{sec:numerical_distance_calculation} describes the histogram-based calculation and clarifies that the plotted TV norm is twice this distance.
Following their comparison, we report the better-performing Prox-sub measure described in Appendix~\ref{app:experiments}. Figure \ref{fig:Comparison_TV} compares MASLA with Prox-sub, P-MALA, and MYULA for $\TV$-$L^2$ sampling in TV distance. As shown in Figure~\ref{fig:Comparison_W2}, all tested algorithms exhibit empirical exponential convergence to the target distribution across the considered $\tau$ values. 

We emphasize that, while exponential ergodicity is observed empirically for these methods in the these literature, we establish it rigorously for MASLA in this paper.

\begin{figure*}[tb]
		\centering
		{\includegraphics[width=0.85\linewidth]{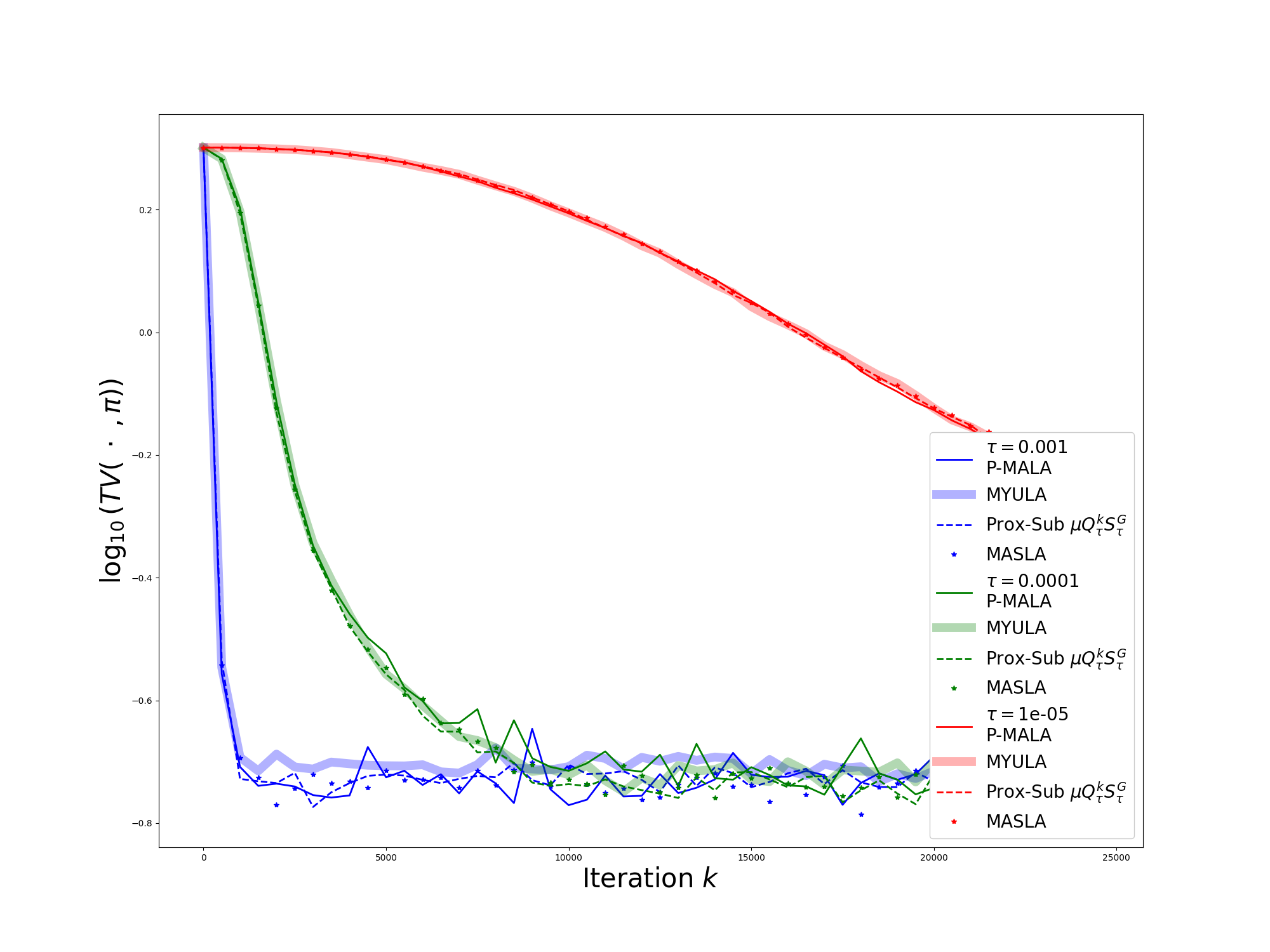}}
		\caption{Comparison of MASLA, Prox-sub \citep{habring2024subgradient}, P-MALA \citep{pereyra2016proximal}, and MYULA \citep{durmus2022proximal} for $\TV-L^2$ sampling in TV distance with $\sigma=1$, $\lambda=5$, and step sizes $\tau \in \{10^{-3}, 10^{-4}, 10^{-5}\}$. Colors represent different $\tau$ values. Thick transparent lines denote MYULA results, thin solid lines denote P-MALA, dashed lines denote Prox-sub, and dotted lines denote MASLA.}
		\label{fig:Comparison_TV}
	\end{figure*}

\subsection{Evaluation beyond the competing methods' convex setting}
\label{sec:main_nonconvex_comparison}
\paragraph{Why this target extends the comparison.}
The first experiment establishes performance in a shared applicability regime. We now examine a target that violates the convexity assumptions used by those baselines and illustrates the difficulty of requiring a single-valued proximal mapping.
Consider $U(x)=|x^2-1|$. This function is locally Lipschitz, since on a bounded interval
$|U(x)-U(y)|\leq|x-y|\,|x+y|$. It has the explicit Clarke subgradient
\[
\partial^c U(x)=
\begin{cases}
\{2x\},&|x|>1,\\
[-2,2],&x=\pm1,\\
\{-2x\},&|x|<1.
\end{cases}
\]
The potential is nonconvex because $U(0)=1>[U(-1)+U(1)]/2=0$. At proximal parameter one and input zero, the minimization objective is
\[
\phi(x)=|x^2-1|+x^2/2=
\begin{cases}1-x^2/2,&|x|\leq1,\\3x^2/2-1,&|x|\geq1.\end{cases}
\]
Both $-1$ and $1$ minimize $\phi$, yielding $\prox_U(0)=\{-1,1\}$. Thus, the convex baseline guarantees do not directly apply, and a proximal construction requiring uniqueness at every parameter needs additional assumptions or a selection rule. While this does not preclude other nonconvex adaptations of proximal methods, MASLA sidesteps this difficulty entirely: by operating directly on the displayed Clarke field, it requires neither solving a proximal subproblem nor imposing conditions to ensure uniqueness of its solution, making it directly applicable in this nonconvex setting.

\paragraph{Isolating the Metropolis correction.}

	The multi-valued nature of the proximal operator hinders  the application of proximal-based methods, however we can still compare with the algorithmic version without the MH correction to numerically visualize the strength of the MH correction. \cite{difonzo2022stochastic} focused on the theoretical analysis of stochastic Langevin differential inclusions of continuous-time while they ran unadjusted Langevin dynamics with the Euler-Maruyama discretization. Here, we refer it as Unadjusted Subdifferential Langevin Algorithm (USLA), which simply draws a subgradient from the set-valued field at each iteration.
%\begin{equation}
%	\label{eqn:USLA}
%	\breve{X}_{n+1} = \breve{X}_n - \gamma \beta(\breve{X}_n) + \sqrt{2 \gamma}\, Z_{n+1}, \qquad \beta(\breve{X}_n) \in D_U(\breve{X}_n),
%\end{equation}
Both chains start at zero and run for 100,000 iterations with step size $0.1$; the first 20\% are discarded as burn-in. Figure~\ref{fig:usla_vs_masla} compares the retained samples with the true density. The reported distance errors and running times are:
\begin{center}
\begin{tabular}{lrrr}
\hline
Algorithm & TV error & $W_2$ error & Time (seconds)\\
\hline
USLA & 0.116761 & 0.092183 & 0.9976\\
MASLA & 0.014363 & 0.008199 & 1.5361\\
\hline
\end{tabular}
\end{center}
Timing was measured on an Apple M1 Pro laptop with 10 cores and 16 GB of memory. In this run, MASLA reduces TV error by about $8.1$ times and Wasserstein error by about $11.2$ times, at about $1.54$ times the runtime of USLA. The comparison therefore quantifies the accuracy benefit and computational cost of correction beyond the convex benchmark.

\begin{figure*}[tb]
		\centering
		{\includegraphics[width=0.7\linewidth]{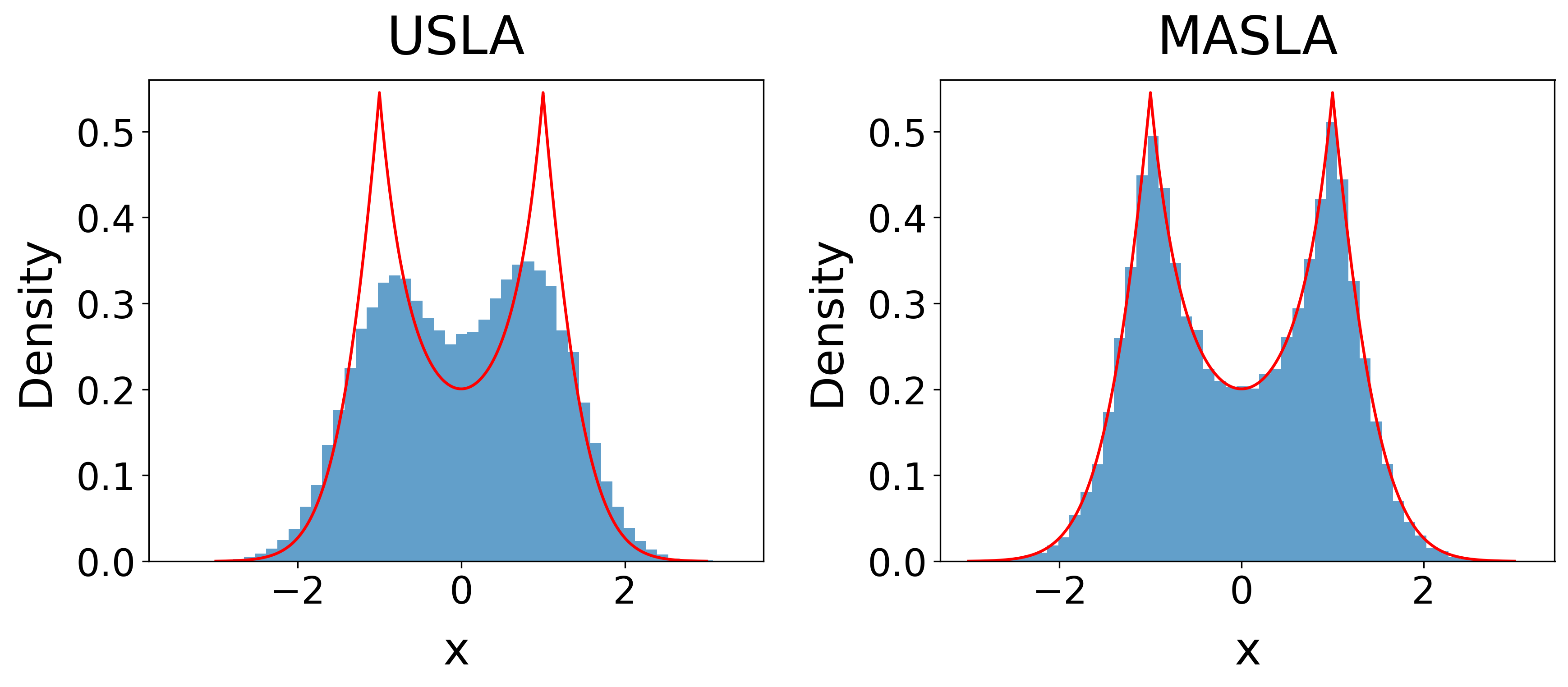}}
		\caption{USLA and MASLA samples compared with the true density. Both chains start at zero and run for 100,000 iterations with step size $\tau=0.1$; the first $20\%$ of iterations are discarded as burn-in.}
		\label{fig:usla_vs_masla}
	\end{figure*}

\subsection{High-dimensional sampling with a quartic potential}
\label{sec:main_high_dimensional_quartic}
To complement the preceding one- and two-dimensional examples, we consider the $100$-dimensional nonsmooth nonconvex target
\[
\pi_{100}(x)\propto\exp\!\left[-\sum_{i=1}^{100}
\left(\frac{x_i^4}{4}-x_i^2+|x_i|\right)\right].
\]
The quartic terms ensure normalizability, the absolute-value terms introduce nonsmooth coordinate hyperplanes, and the negative quadratic terms make the potential nonconvex. We use the selection $\beta_i(x)=x_i^3-2x_i+\operatorname{sign}(x_i)$, with $\operatorname{sign}(0)=0$. Although the target factorizes, proposals update all $100$ coordinates jointly, and MASLA applies a single MH accept--reject decision to the full vector. Appendix~\ref{sec:high_dimensional_quartic} gives the experimental details and Figure~\ref{fig:high_dimensional_quartic} reports the results.

For each step size $\gamma\in\{0.05,0.10,0.20\}$, we run $16$ independent replicates of each algorithm from $0.1\mathbf{1}_{100}$ for up to $10{,}000$ iterations. We record numerical escape of USLA when $\|X_n\|_\infty\geq10^6$ and stop that replicate after the crossing. No USLA replicate crosses the threshold at $\gamma=0.05$ or $0.10$ during the observation period; at $\gamma=0.20$, all $16$ cross it within $336$ iterations. MASLA remains below the threshold in every run, with mean acceptance rates of $58.03\%$, $21.49\%$, and $1.37\%$, respectively. After discarding $2{,}000$ iterations, its first-coordinate histograms agree reasonably with the target marginal at the smaller step sizes, whereas the low acceptance rate at $\gamma=0.20$ leads to poor exploration over the available run length.
This experiment illustrates both the role of the MH correction in controlling explosive excursions from a superlinear drift and the importance of step-size choice in high dimension. The observed threshold crossings provide numerical evidence of USLA instability.

\section{Conclusion}
\label{sec:conclusion}
We introduced MASLA for sampling from distributions with locally Lipschitz potentials that may be nonsmooth and nonconvex. Combining a fixed measurable conservative-field selection with an exact Metropolis--Hastings correction, MASLA preserves the target distribution without requiring smoothing or a proximal mapping. We established reversibility and positive Harris recurrence, characterized selection independence, and rigorously proved geometric ergodicity for positive-definite quadratic potentials with  Lipschitz perturbations admitting bounded conservative fields. Numerical experiments verified the predicted geometric convergence behavior on a convex composite target and demonstrated competitive performance where existing proximal and subgradient methods apply. On a nonsmooth nonconvex target beyond their convex framework, MASLA substantially reduced sampling errors relative to unadjusted subdifferential Langevin dynamics. 
%These results provide a theoretical and computational foundation for MASLA, with convergence analysis for broader potential classes.

\phantomsection
\addcontentsline{toc}{section}{References}
\bibliography{ms_0923_review}
\clearpage
\tableofcontents
\clearpage
\appendix
\section{Additional background and motivation}\label{app:background}
	\label{sec:Introduction}

This appendix reviews the classical MALA construction and its motivation in Section~\ref{sec:MALA}, followed by further motivation for the development of MASLA in Section~\ref{sec:MASLA}.

	\subsection{Metropolis-Adjusted Langevin Algorithm}
	\label{sec:MALA}
	
	The Metropolis-Hastings (MH) algorithm is a fundamental method within the class of Markov Chain Monte Carlo (MCMC) techniques, widely employed for sampling from probability distributions, particularly when the target distribution is known only up to a normalizing constant. Given a target distribution $\pi$ on $\mathbb{R}^d$ with density with respect to (w.r.t.) the Lebesgue measure, which is also denoted as $\pi$ for simplicity,
	\begin{equation}
		\pi(x) = \frac{e^{-U(x)}}{\int e^{-U(y)}\;\mathrm{d} y},
		\label{eq:intro:distribution}
	\end{equation}
	where $U: \mathbb{R}^d \to [0,\infty)$ is the potential function, the MH algorithm enables sampling from $\pi(x)$ even when only the relative density values, i.e., the ratio $\pi(x) / \pi(y)$, are available. Specifically, starting from an initial random variable $X_0$, the MH algorithm constructs a Markov chain $(X_n)_{n \geq 0}$ iteratively through the following steps. At iteration $n$, given the current state $X_n$, a candidate sample $Y_{n+1}$ is generated from a proposal distribution with density $q(X_n, y)$ w.r.t. the Lebesgue measure. This proposed sample is then accepted with probability
	\begin{equation}
		\label{eqn:acceptance_prob}
		\alpha(x,y) = 1 \wedge \frac{\pi(y) q(y,x)}{\pi(x) q(x,y)},
	\end{equation}
	where we use the notation $a \wedge b = \min(a,b)$. If the proposal is accepted, the next state is updated as $X_{n+1} = Y_{n+1}$; otherwise, the chain remains at the current state, i.e., $X_{n+1} = X_n$.
	
	By construction, the Markov chain $(X_n)_{n \geq 0}$ is reversible w.r.t. the target density $\pi$, ensuring that $\pi$ is an invariant distribution. The efficiency of the MH algorithm heavily depends on the choice of the proposal distribution $q$. A common choice is the Gaussian proposal centered at $x \in \mathbb{R}^d$ as
	\begin{equation*}
		q(x,y) = \frac{1}{(2 \sigma^2)^{n/2}} \exp \biggl(-\frac{|x-y|^2}{2 \sigma^2} \biggr).
	\end{equation*}
	For symmetric proposals satisfying $q(x,y) = q(y,x)$, the acceptance probability simplifies to
	\begin{equation*}
		\alpha(x,y) = 1 \wedge \frac{\pi(y)}{\pi(x)}.
		\label{def:acceptanceRWM}
	\end{equation*}
	The Metropolis-Hastings (MH) algorithms utilizing symmetric proposal distributions are specifically referred to as Random Walk Metropolis (RWM) algorithms. To optimize performance, the asymptotic acceptance probability should be tuned to approximately 0.234, as established in \cite{gelman1997weak} and supported by recent studies, such as \cite{ning2025convergence}.

	Intuitively, if the proposal density exploits the structure of the target density, a higher acceptance probability is possible. This leads to the Metropolis-adjusted Langevin Algorithm (MALA), a notable class of algorithms derived from discrete approximations to Langevin diffusions. A Langevin diffusion for 
	$\pi(x)$ is a natural, non-explosive diffusion process that is reversible w.r.t $\pi(x)$. By leveraging the gradient of 
	$\pi(x)$, it tends to move more frequently in directions where 
	$\pi(x)$ increases. As a result, MALA achieves an optimal acceptance probability of approximately 0.574  \citep{roberts1998optimal}, which is significantly higher than that of the RWM algorithm. 	Specifically, assuming $U$ as continuously differentiable, the MALA employs a proposal distribution induced by the following Langevin diffusion:
	\begin{equation}
		\label{eqn:LD}
		d X_t = -\nabla U(X_t)dt + \sqrt{2}d B_t,
	\end{equation}	
	where $B_t$ is a $d$-dimensional Brownian motion. Under additional conditions, this diffusion has a strong solution for any initial point and admits $\pi$ as its unique stationary measure \citep{Roberts1996Exponential}. 
	
	Sampling a path solution of \eqref{eqn:LD} is challenging, so discretizations are often employed to construct a Markov chain with similar long-term behavior. The Euler-Maruyama discretization of \eqref{eqn:LD} is given by	
	\begin{equation}
		\label{eqn:ELD}
		Y_{n+1} = Y_n - \gamma \nabla U(Y_n) + (2 \gamma)^{1/2} Z_{n+1},
	\end{equation}
	for all $n \in \mathbb{N}$, 
	where $\gamma$ is the step size of the discretization, and $(Z_n)_{n \geq 1}$ is an independent and identically distributed (i.i.d.) sequence of $d$-dimensional standard Gaussian random variables.
	The iteration in \eqref{eqn:ELD} is known as the Unadjusted Langevin Algorithm (ULA). According to the classical results on page 342 of \cite{Roberts1996Exponential}: 
	\textit{``whereas the genuinely continuous-time processes often perform well on large classes of target densities, the situation is much more delicate for the approximations which would be used in practice. In particular, naive discretizations of the continuous-time models may lose not only the geometric rates of convergence but also all convergence properties, even for quite standard density $\pi$"}. 
	To address this issue, the MALA incorporates the Markov kernel associated with the Euler-Maruyama discretization \eqref{eqn:ELD} as a proposal kernel within a MH accept-reject scheme. This adjustment ensures that the resulting Markov chain better approximates $\pi$. Specifically, the MALA generates a new Markov chain ${X_n}$ according to 
	\begin{equation} \label{eq:def_MALA} X_{n+1} = X_n +(\breve{Y}_{n+1} - X_n) \mathbbm{1}_{\mathbb{R}_-}(U_{n+1} - \alpha_{\gamma}(X_n, \breve{Y}_{n+1})), 
	\end{equation}
	where $\{U_n\}_{n \in \mathbb{N}}$ is a sequence of i.i.d. uniform random variables on $[0,1]$, 
	\begin{equation}
		\label{eqn:proposal_MALA}
		\breve{Y}_{n+1} = X_n - \gamma \nabla U(X_n)+ (2 \gamma)^{1/2} Z_{n+1}, 
	\end{equation}	
	and $\alpha_{\gamma} : \mathbb{R}^{2d} \to [0,1]$ represents the acceptance probability function defined in \eqref{eqn:acceptance_prob} while explicitly revealing the dependence on the step size $\gamma$.

\subsubsection{Gaussian illustration of discretization instability}
\label{sec:gaussian_motivation_details}
Figure~\ref{fig:mala_vs_ula} uses $U(x)=x^2/(2\sigma^2)$ with $\sigma=0.02$ and $\gamma=0.001$. The corresponding diffusion is an Ornstein--Uhlenbeck process with invariant law $N(0,\sigma^2)$. In contrast, writing $a=1-\gamma/\sigma^2=-1.5$, the ULA recursion \eqref{eq:motivation_gaussian_ula} gives
\[
a^{-n}Y_n=Y_0+\sqrt{2\gamma}\sum_{j=1}^n a^{-j}Z_j.
\]
Since $|a|>1$, the series converges almost surely to a nondegenerate Gaussian random variable, which is nonzero almost surely. Consequently, $|Y_n|\to\infty$ almost surely. This establishes transience analytically, independently of the numerical stopping threshold. The instability depends on the ratio $\gamma/\sigma^2$, rather than on the absolute step size alone; the example does not assert instability for arbitrarily small steps on a fixed Gaussian target.

We simulate 16 replicates of each method from zero, with 100,000 updates and 20,000 burn-in updates. The Python script uses seed 42 and separate NumPy random streams for MALA and ULA, obtained through \texttt{SeedSequence}; each replicate receives independent Gaussian draws and, for MALA, independent acceptance uniforms. The proposal is $y=ax+\sqrt{2\gamma}Z$, with MH acceptance probability
\[
\alpha(x,y)=1\wedge\exp\!\left(\frac{\gamma}{4\sigma^4}(x^2-y^2)\right).
\]
Each ULA path stops at its first crossing of $|x|/\sigma\geq10^6$, retaining the crossing value for the trajectory plot without clipping or restarting. Panels (a)--(b) display the first replicate, fixed in advance: MALA's first 10,000 updates and ULA's path through iteration 41. Panels (c)--(d) use 64 equal-width bins on $[-4\sigma,4\sigma]$, dividing bin counts by bin width and the total planned number of retained draws, $16\times80,000$. Thus, stopped paths are not silently omitted from the normalization. All ULA paths stop within 31--41 updates, leaving no post-burn-in observations. The MALA histogram pools all retained draws and has standardized variance $1.0019$; its mean acceptance rate across the full runs is $39.6\%$. These finite-run results demonstrate accurate sampling from this initialization, without implying rapid convergence from arbitrary starting points.

	\subsection{Metropolis-Adjusted Subdifferential Langevin Algorithm}
	\label{sec:MASLA}
	
	In this paper, we introduce the Metropolis-adjusted subdifferential Langevin algorithm (MASLA), a generalization of MALA that incorporates nonsmooth, nonconvex potential functions. In the realm of non-convex and non-smooth optimization, we consider generalized derivatives in the name of conservative field, which was introduced in \cite{bolte2021conservative} in recent years. Conservative fields coincide with classical gradients almost everywhere, having the Clarke subgradient \citep{clarke1990optimization} as the minimal convex conservative field. If $f$ is differentiable, $\nabla f$ is of course a conservative field. Consider the example from \cite{difonzo2022stochastic} with $f: \mathbb{R} \to \mathbb{R}$ and its conservative field as follows:
	\[
	f(x) = \begin{cases} 
		-x - 1, & x < -1 \\ 
		x + 1, & -1 \leq x < 0 \\ 
		1 - x, & 0 \leq x < 1 \\ 
		x - 1, & x \geq 1 
	\end{cases}, \quad
	D_f(x) = \begin{cases} 
		-1, & x < -1 \\ 
		[-1, 1], & x = -1 \\ 
		1, & -1 < x < 0 \\ 
		[-1, 1], & x = 0 \\ 
		-1, & 0 < x < 1 \\ 
		[-1, 1], & x = 1 \\ 
		1, & x > 1 
	\end{cases}.
	\]
	A detailed explanation of a simple locally Lipschitz function, lacking a well-defined proximal operator at all points but possessing an explicit Clarke subgradient, is provided in Section \ref{sec:nonapplicable}.
	
	Optimization plays a fundamental role in both statistics and machine learning, as many methods in these fields are formulated as optimization problems. In statistics, optimization is essential for parameter estimation, model fitting, and inference, while in machine learning, it underpins the training of models by minimizing loss functions and improving predictive performance.  To better capture learning and prediction tasks, structural constraints like sparsity or low rank are often imposed, or the objective is designed to be non-convex. This is particularly relevant for non-linear models such as tensor models and deep networks. While non-convex optimization offers immense modeling flexibility, the resulting problem is often NP-hard to solve. This challenge is common in fields like genomics and signal processing, where \textit{data starvation} occurs. Consider a set of n covariate/response pairs $(x_1,y_1),\ldots,(x_n,y_n)$, where $x_i\in \mathbb{R}^d$ and $y \in \mathbb{R}$, \textit{data starvation} means that the number of data points $n$ is much smaller than the number of model parameters $p$, i.e., $n \ll p$. Standard statistical methods, which require $n \geq p$, struggle in such scenarios. \textit{Sparse recovery} addresses this issue by fitting a sparse model vector to the data:
	\begin{align*}
		\hat{w}_\text{sp} = \underset{w \in \mathbb{R}^d}{\arg\min} 
		\sum_{i=1}^n(y_i - x_i^Tw)^2, 
	\end{align*}
 such that $w \in \mathcal{B}_0(s):=\{x\in \mathbb{R}^d: |\operatorname{supp}(x)|\leq s\}$, the set of s-sparse vectors with $\operatorname{supp}(x):=\{i:x_i\neq 0\}$. 
	Although the objective function is convex, the constraint is non-convex. For detailed illustrations and examples, see \cite{jain2017non}. In the MCMC regime, stochastic gradient Langevin dynamics have been developed for non-convex learning; see, e.g., \cite{deng2020non}.
	
	Modern problems in AI and numerical analysis often require nonsmooth formulations, for example, deep learning models employing ReLU activation functions and square loss. In statistical $M$-estimation, nonsmooth regularization functions, particularly the $\ell_1$-norm and its variants, have long served as foundational tools, demonstrating strong empirical and theoretical performance over decades. In the Bayesian framework, nonsmooth priors also play a central role in high-dimensional models~\citep{seeger2008bayesian,ohara2009review}. The Bayesian analogue of the $\ell_1$-penalty is the Laplace prior, which has been widely studied~\citep{park2008bayesian,carvalho2008high,dalalyan2018exponentially}.
	To address the challenges of nonsmooth potentials in Bayesian computation, \cite{pereyra2016proximal} proposed applying both MALA and ULA to the Moreau–Yosida envelope of a nonsmooth potential $g$. \cite{bernton2018langevin} analyzed the performance of the unadjusted approach in a specific setting. \cite{durmus2018efficient} extended these methods to composite potentials of the form $f + g$, where $g$ is nonsmooth, by replacing $g$ with its smooth Moreau–Yosida approximation. This smoothing approach has also been applied to Hamiltonian Monte Carlo~\citep{chaari2016hamiltonian}.
	Later \citet{durmus2019analysis} studied the mixing time problem for nonsmooth composite objectives using Wasserstein gradient flow. \citet{habring2024subgradient} further investigated the setting where the target potential takes the form $U(x) = F(x) + G(Kx)$, with $F: \mathbb{R}^d \to [0,\infty)$ and $G: \mathbb{R}^{d'} \to [0,\infty)$ both proper, convex, and lower semi-continuous, and $K: \mathbb{R}^d \to \mathbb{R}^{d'}$ a linear operator. 
	
	\newpage
	\section{Mathematical background and full proofs}\label{app:proofs}
	\label{sec:Main_Results}

This appendix develops the conservative-field framework in Section~\ref{sec:Preliminaries}, establishes the well-posedness and exact invariance of MASLA in Section~\ref{sec:Wellposedness}, and provides the selection-independence and convergence proofs in Section~\ref{sec:Convergence}.
	
	\subsection{Conservative fields and generalized derivatives}
	\label{sec:Preliminaries}
	
	In this paper, we consider locally Lipschitz continuous functions in the Euclidean space  $\mathbb{R}^d$, which is equipped with the canonical Euclidean inner product $\left\langle \cdot,\cdot\right\rangle$ and its induced norm $\|\cdot\|$. Denote by $\mathcal{B}(\RR^d)$ the Borel $\sigma$-field on $\RR^d$. 
	The notation $A^c$ represents the complementary set of a set $A$ and $\operatorname{cl}(A)$ denote the closure of $A$. Let $\lambda$ be the Lebesgue measure on $\RR$ and $\RR^d$, regardless of the dimension.
	
	\begin{definition}[Closed graph]
		The graph of a set-valued map $D\colon \RR^d \rightrightarrows \RR^q$ is given by
		\begin{align*}
			\operatorname{Graph}(D) = \Big\{ (x,z):\, x \in \RR^d,\, z \in D(x)\Big\}.
		\end{align*}
		$D$ is said to have closed graph or to be graph closed if $\operatorname{Graph}(D)$ is closed as a subset of $\RR^{d + q}$. 
	\end{definition}
	
		\begin{definition}[Upper semicontinuous]
		We say that a set-valued map $D\colon \RR^d \rightrightarrows \RR^q$ is upper semicontinuous, if its graph is closed. 
	\end{definition}

	\begin{definition}[Absolutely continuous curve]
		An absolutely continuous curve is a continuous function $x \colon \RR \to \RR^d$ which admits a derivative $\dot{x}$ for Lebesgue almost all $t \in \RR$ in which case $\dot{x}$ is Lebesgue measurable, and $x(t) - x(0)$ is the Lebesgue integral of $\dot{x}$ between $0$ and $t$ for all $t \in \RR$.  
	\end{definition}

	\begin{definition}[Conservative set-valued fields]
		\label{def:conservative}
		Let $D \colon \RR^d \rightrightarrows \RR^d$ be a set-valued map. $D$ is a {\em conservative (set-valued)  field} whenever it has closed graph, nonempty compact values and for any absolutely continuous loop $\gamma \colon [0,1] \to \RR^d$, that is   $\gamma(0) = \gamma(1)$, we have
		\begin{align*}
			\int_0^1 \max_{v \in D(\gamma(t))} \left\langle \dot{\gamma}(t), v \right\rangle \dd t = 0,
		\end{align*}
		which is equivalent to 
		\begin{align*}
			\int_0^1 \min_{v \in D(\gamma(t))} \left\langle \dot{\gamma}(t), v \right\rangle \dd t = 0,
		\end{align*}
		where the integrals are understood in the Lebesgue sense. 
		Another equivalent characterization in term of the Aumann integral is given by
		\begin{equation}\label{circul0}
			\int_0^1 \left\langle  D(\gamma(t)), \dot{\gamma}(t)\right\rangle \dd t = \{0\}.
		\end{equation}
	\end{definition}
	
		Chain rule proved in Lemma 2 of \cite{bolte2021conservative} characterizes conservativity in the following sense:
	\begin{lemma}[Chain rule and conservativity, \cite{bolte2021conservative}]
		\label{lem:chainRuleAC}
		Let $D \colon \RR^d \rightrightarrows \RR^d$ be a locally bounded, graph closed set-valued map and $f\colon \RR^d \to \RR$ a locally Lipschitz continuous function. Then $D$ is a conservative field for $f$, if and only if for any absolutely continuous curve $x \colon [0,1] \to \RR^d$, the function $t \mapsto f(x(t))$ satisfies 
		\begin{align}
			\frac{d}{d t} f(x(t)) = \left\langle v, \dot{x}(t) \right\rangle\qquad \forall v \in D(x(t)),
			\label{eq:chainRuleAC}
		\end{align}
		for almost all $t \in [0,1]$.
	\end{lemma}

	\begin{definition}[Potential functions of conservative fields]
		\label{def:conservativeMapForF}
		Let $D \colon \RR^d \rightrightarrows \RR^d$ be a conservative field. 
		Define a function $f$ by any of the equivalent expressions 
		\begin{eqnarray*}
			f(x)& = & f(0)+\int_0^1 \max_{v \in D(\gamma(t))} \left\langle \dot{\gamma}(t), v \right\rangle \dd t \label{pot1}\\
			&= & f(0)+\int_0^1 \min_{v \in D(\gamma(t))} \left\langle \dot{\gamma}(t), v \right\rangle \dd t \label{pot2}\\
			& = & f(0)+\int_0^1 \left \langle \dot{\gamma}(t), D(\gamma(t)) \right\rangle \dd t,
			\label{pot3}\end{eqnarray*}
		where $\gamma \colon [0,1] \to \RR^d$ is any absolutely continuous function with $\gamma(0) = 0$ and $\gamma(1)=x$.  The function 
$f$ is well defined and unique up to an additive constant. It is called a potential function for 
$D$. We also say that 
$D$ admits
$f$ as a potential, or equivalently, that 
$D$ is a conservative field for 
$f$.
	\end{definition}
	
	Recall that the convex hull of a set $ S $ in a real vector space (e.g., $\mathbb{R}^d$) is the smallest convex set that contains $ S $, denoted as $\conv(S)$.
	\begin{proposition}[\cite{bolte2021conservative}, Remark 3]
		The following properties hold:
		\begin{itemize}
			\item The potential function $f$ defined in Definition \ref{def:conservativeMapForF} is  locally Lipschitz continuous.
			\item If $D_1,D_2$ are two graph closed set-valued mappings with compact nonempty values, then $D_1\subset D_2$ and $D_2$ conservative implies that $D_1$ is conservative as well.
			\item If $D$ is conservative, $\conv(D)$ is conservative as well.
		\end{itemize}
	\end{proposition}

	A locally Lipschitz continuous function, $f \colon \RR^d \to \RR$ is differentiable almost everywhere by Rademacher's theorem \citep{evans2015measure}. 
	
	\begin{definition}[Clarke subgradient, \cite{clarke2008nonsmooth}]
		\label{def:Clarke}
		For any locally Lipschitz continuous function
		$f:\RR^d \rightarrow \RR$, we denote by $\cD_f$ the set of points $x\in
		\RR^d$ such that $f$ is differentiable at $x$, and then the Clarke subgradient of $f$ is given for any $x \in \RR^d$, by
		\begin{equation}
			\label{eqn:Clarke}
			\partialc  f (x) = \conv \Big\{ y\in \RR^d\,:\,\exists (x_n)_{n\in \NN}\in \cD_f^\NN\ \text{ s.t. } (x_n,\nabla f(x_n)) \rightarrow (x,y)\Big\}\,.
		\end{equation}
	\end{definition}
	
	The following proposition collects the results of Theorem 1, Corollary 1, and Proposition 1 of \cite{bolte2021conservative}.
	\begin{proposition}[\cite{bolte2021conservative}]
		\label{bolte2021conservative_prop}
		Consider a conservative field $D \colon \RR^d \rightrightarrows \RR^d$ for $f \colon \RR^d \to \RR$. Then the following hold:
		\begin{enumerate}%[label=(\roman*)]
			\item[(i)] $D=\{\nabla f\}$ Lebesgue almost everywhere. 
			\item[(ii)] $\partialc f$ is a conservative field for $f$, and $\partialc f (x) \subset \conv(D(x))$ for all $x \in \RR^d$. 
			\item[(iii)] We say that $x$ is  $D$-critical for $f$ if $D(x)\ni 0$, and the value $f(x)$ is then called a $D$-critical value. When $D \colon \RR^d \rightrightarrows \RR^d$ has non-empty compact values and closed graph, for $x \in \RR^d$ being a local minimum or local maximum of $f$, we have $0 \in \conv(D(x))$.
		\end{enumerate}
	\end{proposition}
	
	\begin{proposition}[\cite{bolte2021conservative}, Corollary 2]
		Let $f \colon \RR^d \to \RR$ be locally Lipschitz continuous, then the following assertions are equivalent:
		\begin{itemize}
			\item[(i)] $f$ is the potential of a conservative field on $\RR^d$.
			\item[(ii)] $\partial^cf$ is a conservative field. 
			\item[(iii)] $f$ has chain rule for the Clarke subgradient.
		\end{itemize} 
	\end{proposition}
The breadth of the class of potential functions for conservative fields is highlighted by the following sufficient conditions.
	\begin{proposition}[\cite{bolte2021conservative}, Proposition 2]
	Let $f:\RR^d\to \RR$ be locally Lipschitz continuous, the following are sufficient conditions for $f$ to be potential of a conservative field on $\RR^d$:
	\begin{enumerate}%[label=(\roman*)]
		\item[(i)]  $f$  is convex or concave.
		\item[(ii)]  $f$ or $-f$  is Clarke regular.
		\item[(iii)]  $f$ or $-f$  is prox regular.
		\item[(iv)]  $f$ is real semialgebraic or more generally tame.
	\end{enumerate}	
\end{proposition}
Tame functions are defined via $o$-minimal structures, which we recall in the following definition. The simplest  $o$-minimal structure is given by the class of
real  semialgebraic objects.
Recall that a set $A \subset \RR^d$ is called {\em semialgebraic} if it is a finite union of sets of the form $$\displaystyle  \bigcap_{i=1}^k \Big\{x \in \RR^d \mid g_{i}(x) < 0, \; h_{i}(x) = 0 \Big\},$$
where the functions $g_{i}, h_{i}: \RR^d \to \RR$ are real polynomial functions and $k\geq 1$.
	\begin{definition}[$o$-minimal structure, \cite{Coste1999}]
		An $o$-minimal structure on $(\RR,+,\cdot)$ is a collection of sets
		$\mathcal{O} = (\mathcal{O}_{d})_{d \in \NN}$ where each $\mathcal{O}_{d}$ is itself a family of
		subsets of $\RR^d$, such that for each $d \in \NN$:
		\begin{enumerate}
			\item[(i)] $\mathcal{O}_{d}$ is stable by complementation, finite union, finite intersection and contains $\RR^d$.
			\item[(ii)]  if $A$ belongs to $\mathcal{O}_{d}$, then $A \times \RR$ and $\RR \times A$
			belong to $\mathcal{O}_{d+1}$;
			\item[(iii)]  if $\pi: \RR^{d+1} \to \RR^d$ is the canonical projection onto $\RR^d$ then,
			for any $A \in \mathcal{O}_{d+1}$, the set $\pi(A)$ belongs to $\mathcal{O}_{d}$;
			\label{it:algebraic}
			\item[(iv)]  $\mathcal{O}_{d}$ contains the family of real algebraic subsets of $\RR^d$, that is,
			every set of the form
			$\{ x \in \RR^d \mid g(x) = 0 \}$
			where $g: \RR^d \to \RR$ is a polynomial function;
			\item[(v)]  the elements of $\mathcal{O}_1$ are exactly the finite unions of  intervals.
		\end{enumerate}
	\end{definition}
	
	\begin{definition}[Tame functions, \cite{van1998tame}]
		\label{def:tame}
		A subset of $\RR^d$ which belongs to an $o$-minimal structure $\mathcal{O}$ is said to be
		definable in $\mathcal{O}$. 
		A set-valued function is said to be definable in $\mathcal{O}$ whenever its graph is definable in $\mathcal{O}$. The terminology tame refers to definability in an $o$-minimal structure. 
	\end{definition}

	\subsection{Wellposedness of MASLA}
	\label{sec:Wellposedness}
In this section, we demonstrate that MASLA is a well-designed algorithm with desirable theoretical properties. 
We adopt the term \emph{well-posedness} from the differential equations literature, since MASLA is based on a set-valued extension of the Langevin diffusion \eqref{eqn:SLD}. 
We begin by showing that \eqref{eqn:SLD} admits a solution and thus defines a properly posed stochastic differential inclusion (SDI).
A function $f:\mathbb{R}^d \to \mathbb{R}^d$ is said to be \emph{piecewise Lipschitz continuous} if there exists a hypersurface $\Theta$, referred to as the \emph{exceptional set} of $f$, such that $\Theta$ has finitely many connected components $\{\theta_i\}_{i=1}^m$ and the restriction
$f\vert_{\mathbb{R}^d \setminus \Theta}$ is intrinsically Lipschitz; see Definition 3.2 of \cite{LeobacherSzolgyenyi2017}.
Similarly, a set-valued map $F:\mathbb{R}^d \rightrightarrows \mathbb{R}^d$ is called \emph{piecewise Lipschitz continuous} if there exists an exceptional set $\Theta$, such that $F$ is Lipschitz continuous on $\mathbb{R}^d \setminus \Theta$ w.r.t. the Hausdorff metric; see \cite{baierFarkhi2013}.
By Theorem~15 of \cite{difonzo2022stochastic}, if the conservative field $D_U$ is piecewise Lipschitz continuous with a $C^3$ exceptional set $\Theta$ of positive reach, then the SDI \eqref{eqn:SLD} admits a strong solution that is unique almost everywhere. Notably, the Clarke subgradient of a continuous tame function (Definition~\ref{def:tame}) satisfies these assumptions; we refer interested readers to \cite{difonzo2022stochastic} for further examples. Moreover, by Corollary~1 of \cite{bolte2021conservative}, restated here as Proposition~\ref{bolte2021conservative_prop}, the Clarke subgradient is the minimal convex conservative field.

Independently of these SDI existence results, Proposition~\ref{thm:stationary} establishes that the MASLA kernel is well-defined, admits $\pi$ as an invariant distribution, and is reversible w.r.t. $\pi$ for every step size $\gamma>0$, using only a fixed Borel measurable selection and integrability of the target density.

\begin{restatedresult}{Example~\ref{example:special_case}}
Let $g:\mathbb{R}\to\mathbb{R}$ be defined by
	\[
	g(x) :=
	\begin{cases}
		x^{2}\sin(1/x), \quad& x\neq 0,\\
		0, & x=0.
	\end{cases}
	\]
	Then $g$ is locally Lipschitz on $\mathbb{R}$ and differentiable at $x=0$ with
	\[
	\nabla g(0)=0.
	\]
	However, its Clarke subdifferential at $0$ satisfies
	\[
	\partial g(0)=[-1,1].
	\]
\end{restatedresult}

\begin{proof}[Proof of Example \ref{example:special_case}]
	Since $g$ is continuously differentiable on $\mathbb{R}\setminus\{0\}$ and bounded by
	$|g(x)|\le |x|^{2}$, it is locally Lipschitz on $\mathbb{R}$.
	Differentiability at $0$ follows from
	\[
	\nabla g(0)
	= \lim_{h\to 0}\frac{g(h)-g(0)}{h}
	= \lim_{h\to 0} h\sin(1/h)
	=0.
	\]
	
	For $x\neq 0$, the derivative of $g$ is
	\[
	g'(x)=2x\sin(1/x)-\cos(1/x).
	\]
	As $x\to 0$, the term $2x\sin(1/x)$ converges to $0$, while $\cos(1/x)$ oscillates between
	$-1$ and $1$. Hence, for any $a\in[-1,1]$, there exists a sequence $(x_n)_{n\ge 1}$
	converging to $0$ such that
	\[
	g'(x_n)\to a.
	\]	
	By definition of the Clarke subdifferential \eqref{eqn:Clarke},
	\[
	\partial g(0)
	=\conv\Big\{ \lim_{n\to\infty} g'(x_n)
	:\ x_n\to 0,\ g \text{ differentiable at } x_n \Big\}
	= [-1,1].
	\]
\end{proof}

\begin{restatedresult}{Proposition~\ref{thm:stationary} (Well-posedness and exact invariance)}
Let $U:\mathbb R^d\to\mathbb R$ be finite-valued and Borel measurable and satisfy 
\begin{equation*}
	0<Z_U:=\int e^{-U(x)}\,\lambda(dx)<\infty,
	\qquad \pi(x)=Z_U^{-1}e^{-U(x)}.
\end{equation*}
Fix a finite-valued Borel map $\beta:\mathbb R^d\to\mathbb R^d$ (in MASLA, a selection of $D_U$). For every $\gamma>0$, define
\begin{equation*}
	q_\gamma(x,y)=(4\pi\gamma)^{-d/2}
	\exp\left(-\frac{\|y-x+\gamma\beta(x)\|^2}{4\gamma}\right),
\end{equation*}
\begin{align*}
	\alpha_\gamma(x,y)&=1\wedge\frac{\pi(y)q_\gamma(y,x)}{\pi(x)q_\gamma(x,y)},
	& r_\gamma(x)&=1-\int q_\gamma(x,y)\alpha_\gamma(x,y)\,\lambda(dy).
\end{align*}
Then the Markov transition kernel
\begin{equation*}
	R_\gamma(x,A)=\int_A q_\gamma(x,y)\alpha_\gamma(x,y)\,\lambda(dy)+r_\gamma(x)\delta_x(A),
	\quad A\in\mathcal B(\mathbb R^d),
\end{equation*}
is well-defined from every initial law, reversible with respect to $\pi$, and admits $\pi$ as an invariant probability distribution.
\end{restatedresult}

\begin{proof}[Proof of Proposition~\ref{thm:stationary}]
Fix $\gamma>0$. Conditional on $X_n=x$, the proposal is Gaussian with
mean $x-\gamma\beta(x)$ and covariance $2\gamma I_d$, which gives
\eqref{eq:masla_proposal_density}. Since $\beta$ is Borel measurable and
finite-valued, $q_\gamma$ is jointly Borel measurable, strictly positive,
and finite, and $\int q_\gamma(x,y)\,\lambda(dy)=1$ for every $x$.
Likewise, $\pi$ is Borel measurable, strictly positive, and finite.
Consequently, the ratio in \eqref{eq:masla_acceptance} is well-defined
for every pair $(x,y)$, and $\alpha_\gamma$ is Borel measurable with
$0<\alpha_\gamma(x,y)\leq 1$.

Measurability of parameter integrals shows that $r_\gamma$ is Borel
measurable, with $0\leq r_\gamma(x)\leq 1$. For every Borel set $A$,
$x\mapsto R_\gamma(x,A)$ is thus Borel measurable. For each $x$,
$R_\gamma(x,\cdot)$ is a nonnegative measure of total mass one by
\eqref{eq:masla_acceptance}. This proves that $R_\gamma$ is a Markov
kernel and defines the MASLA chain for any initial law.

To prove reversibility, observe that for every $x,y\in\mathbb{R}^d$,
\begin{align*}
\pi(x)q_\gamma(x,y)\alpha_\gamma(x,y)
&=\min\{\pi(x)q_\gamma(x,y),\pi(y)q_\gamma(y,x)\}\\
&=\pi(y)q_\gamma(y,x)\alpha_\gamma(y,x).
\end{align*}
The rejection term is also symmetric: for any Borel sets $A,B$,
\[
\int_A \pi(x)r_\gamma(x)\delta_x(B)\,\lambda(dx)
=\int_{A\cap B}\pi(x)r_\gamma(x)\,\lambda(dx).
\]
More explicitly, write $\pi(dx)=\pi(x)\lambda(dx)$ and expand the
kernel in \eqref{eq:masla_kernel}. For any Borel sets $A,B$, we obtain
\begin{align*}
&\int_A \pi(dx)R_\gamma(x,B)\\
&=\int_A\int_B \pi(x)q_\gamma(x,y)\alpha_\gamma(x,y)
\,\lambda(dy)\lambda(dx)+\int_{A\cap B}\pi(x)r_\gamma(x)\,\lambda(dx)\\
&=\int_A\int_B \pi(y)q_\gamma(y,x)\alpha_\gamma(y,x)
\,\lambda(dy)\lambda(dx)+\int_{A\cap B}\pi(y)r_\gamma(y)\,\lambda(dy)\\
&=\int_B \pi(y)\left[\int_A q_\gamma(y,x)\alpha_\gamma(y,x)
\,\lambda(dx)\right]\lambda(dy)+\int_B\pi(y)r_\gamma(y)\delta_y(A)\,\lambda(dy)\\
&=\int_B \pi(dy)R_\gamma(y,A).
\end{align*}
The second equality uses the pointwise symmetry of the accepted-move
term proved above and renames the integration variable in the rejection
term. The third equality follows from Tonelli's theorem, since the
double-integral integrand is nonnegative, and from
$\delta_y(A)=\mathbbm{1}_A(y)$. The last equality is precisely the
definition of $R_\gamma(y,A)$. This establishes detailed balance.
Taking $A=\mathbb{R}^d$ yields
$\int \pi(dx)R_\gamma(x,B)=\pi(B)$, proving invariance.
Since $\gamma>0$ was arbitrary, all conclusions hold for every positive
step size.
\end{proof}
	
\subsection{Convergence of MASLA}
\label{sec:Convergence}
\begin{restatedresult}{Proposition~\ref{prop:selection_independence} (Selection independence)}
Suppose $U$ is locally Lipschitz and satisfies \eqref{eq:standing_target}, and $D_U$ is a conservative field for $U$. 
Let $N_U$ be its nondifferentiability set. Consider finite-valued Borel selections $\beta_1,\beta_2$ of $D_U$ with
\[
\beta_1(x)=\beta_2(x)=\nabla U(x),\qquad x\notin N_U,
\]
with possibly different values on $N_U$. For $i\in\{1,2\}$ and
$\gamma>0$, let $R_{\gamma,i}$ be the MASLA kernel constructed from
$\beta_i$ as in \eqref{eq:masla_proposal_density}--\eqref{eq:masla_kernel}.
Then, for every $\gamma>0$,
\begin{equation*}
R_{\gamma,1}(x,\cdot)=R_{\gamma,2}(x,\cdot),
\qquad x\notin N_U.
\end{equation*}
In particular, the kernels agree for $\lambda$-almost every $x$ and
for $\pi$-almost every $x$.

If both chains have the same initial law $\mu$ with $\mu(N_U)=0$,
then their laws on the sequence space $(\mathbb{R}^d)^{\{0,1,\ldots\}}$
coincide, and each chain satisfies
\[
\mathbb{P}_\mu\bigl(X_n\notin N_U\text{ for every }n\geq0\bigr)=1.
\]
This applies, in particular, to every absolutely continuous initial law
and to every deterministic initial point outside $N_U$. Moreover, an
absolutely continuous initial law remains absolutely continuous at every
iteration under either kernel.
\end{restatedresult}

\begin{proof}[Proof of Proposition~\ref{prop:selection_independence}]
The nondifferentiability set $N_U$ is Borel, and Rademacher's theorem
gives $\lambda(N_U)=0$. Since $\pi$ has a Lebesgue density, also
$\pi(N_U)=0$. Fix $\gamma>0$, and write $q_{\gamma,i}$,
$\alpha_{\gamma,i}$, and $r_{\gamma,i}$ for the proposal density,
acceptance probability, and rejection probability associated with
$\beta_i$. Set
\[
p_{\gamma,i}(x,y):=q_{\gamma,i}(x,y)\alpha_{\gamma,i}(x,y)
=\min\left\{q_{\gamma,i}(x,y),
\frac{\pi(y)}{\pi(x)}q_{\gamma,i}(y,x)\right\}.
\]

For a fixed $x\notin N_U$, equality of the drifts at $x$ implies
$q_{\gamma,1}(x,y)=q_{\gamma,2}(x,y)$ for every $y$.
For $y\notin N_U$, equality of the drifts at $y$ also implies
$q_{\gamma,1}(y,x)=q_{\gamma,2}(y,x)$. Thus
\[
p_{\gamma,1}(x,y)=p_{\gamma,2}(x,y),
\qquad x,y\notin N_U.
\]
The reverse proposal density can differ when $y\in N_U$, but this
set is Lebesgue null. Consequently, for every Borel set $A$ and every
$x\notin N_U$,
\[
\int_A p_{\gamma,1}(x,y)\,\lambda(dy)
=\int_A p_{\gamma,2}(x,y)\,\lambda(dy).
\]
Taking $A=\mathbb{R}^d$ also gives
$r_{\gamma,1}(x)=r_{\gamma,2}(x)$. Substituting these identities into
\eqref{eq:masla_kernel} proves \eqref{eq:selection_kernel_equality},
including equality of the rejection atoms.

Next, fix either $i$ and a current state $x\notin N_U$. Applying
\eqref{eq:masla_kernel} to the Borel set $N_U$ gives
\[
R_{\gamma,i}(x,N_U)
=\int_{N_U}p_{\gamma,i}(x,y)\,\lambda(dy)
+r_{\gamma,i}(x)\mathbbm{1}_{N_U}(x).
\]
Both terms vanish separately. For the accepted-move term, the bound
$0\leq p_{\gamma,i}(x,y)\leq q_{\gamma,i}(x,y)$ and
$\lambda(N_U)=0$ imply
\[
0\leq\int_{N_U}p_{\gamma,i}(x,y)\,\lambda(dy)
\leq\int_{N_U}q_{\gamma,i}(x,y)\,\lambda(dy)=0.
\]
Indeed, the Gaussian proposal has a Lebesgue density, so it assigns
zero probability to the Lebesgue-null set $N_U$; restricting to accepted
proposals cannot increase this probability. For the rejection term,
$\delta_x(N_U)=\mathbbm{1}_{N_U}(x)=0$ because $x\notin N_U$:
rejection leaves the chain at $x$, which is still outside $N_U$.
Thus $R_{\gamma,i}(x,N_U)=0$ for every $x\notin N_U$.

To make the induction explicit, let $\mu_{n,i}$ denote the law of the
chain at iteration $n$ under $R_{\gamma,i}$. If $\mu_{n,i}(N_U)=0$,
then the Markov property gives
\begin{align*}
\mu_{n+1,i}(N_U)
&=\int_{N_U^c}R_{\gamma,i}(x,N_U)\,\mu_{n,i}(dx)
+\int_{N_U}R_{\gamma,i}(x,N_U)\,\mu_{n,i}(dx)\\
&=0.
\end{align*}
The first integral is zero by the pointwise identity just proved, and
the second is zero because $\mu_{n,i}(N_U)=0$ and the kernel takes
values in $[0,1]$. Starting from $\mu_{0,i}=\mu$ with $\mu(N_U)=0$,
induction therefore gives
$\mathbb{P}_\mu(X_n\in N_U)=0$ for each $n$. A countable union of these
null events shows that the chain avoids $N_U$ at all iterations almost
surely. Since the two kernels coincide at every point of $N_U^c$,
iteration of the kernels from the same initial law gives identical
finite-dimensional distributions. Cylinder sets generate the product
sigma-algebra, so the two laws on sequence space coincide.

Finally, let $\eta\ll\lambda$ and let $A$ be any Lebesgue-null Borel
set. For either kernel,
\begin{align*}
(\eta R_{\gamma,i})(A)
&=\int_{\mathbb{R}^d}\eta(dx)
\int_A p_{\gamma,i}(x,y)\,\lambda(dy)
+\int_A r_{\gamma,i}(x)\,\eta(dx)\\
&=0.
\end{align*}
To justify the vanishing of the first term explicitly, the Gaussian
proposal density satisfies
\[
0\leq p_{\gamma,i}(x,y)
=q_{\gamma,i}(x,y)\alpha_{\gamma,i}(x,y)
\leq q_{\gamma,i}(x,y)\leq(4\pi\gamma)^{-d/2}
\]
for every $x,y$, since $\alpha_{\gamma,i}(x,y)\leq1$ and the
exponential factor in the Gaussian density is at most one. Therefore,
for each fixed $x$,
\[
0\leq\int_A p_{\gamma,i}(x,y)\,\lambda(dy)
\leq(4\pi\gamma)^{-d/2}\lambda(A)=0.
\]
The outer integral with respect to $\eta(dx)$ thus integrates the
identically zero function. For the second term, absolute continuity
$\eta\ll\lambda$ implies $\eta(A)=0$, and the rejection probability
satisfies $0\leq r_{\gamma,i}(x)\leq1$. Hence
\[
0\leq\int_A r_{\gamma,i}(x)\,\eta(dx)
\leq\int_A1\,\eta(dx)=\eta(A)=0.
\]
These bounds show explicitly that neither density weighting in the
first term nor rejection-probability weighting in the second can give
positive mass to a set that is null for the respective integrating
measure. Hence $\eta R_{\gamma,i}\ll\lambda$, and induction proves
preservation of absolute continuity.
\end{proof}

\begin{remark}
The requirement that both selections equal $\nabla U$ at every
differentiability point is explicit here; membership in $D_U(x)$ alone
does not guarantee this pointwise equality, as discussed in
Remark~\ref{remark:special_case}. Also, the proposition does not assert
selection independence for an initial point in $N_U$: the drift at such
a point can affect the proposal, acceptance, and rejection probabilities.
For initial laws that give $N_U$ zero mass, however, arbitrary changes
of the selection on $N_U$ have no effect on the law of the MASLA chain.
\end{remark}

\begin{restatedresult}{Proposition~\ref{prop:harris_recurrence} (Irreducibility and Harris recurrence)}
Under conditions of Proposition~\ref{thm:stationary}, including normalizability and a finite-valued Borel map $\beta$, the MASLA kernel $R_\gamma$ is
$\lambda$-irreducible and $\pi$-irreducible for every $\gamma>0$. More precisely, $R_\gamma(x,A)>0$ for every
$x\in\mathbb{R}^d$ and every Borel set $A$ with $\lambda(A)>0$.
Furthermore, the chain is positive Harris recurrent: for every such
set $A$ and every initial point $x$,
\[
\mathbb{P}_x\bigl(X_n\in A\text{ for infinitely many }n\bigr)=1,
\]
and it admits an invariant probability distribution. Its unique
invariant probability distribution is $\pi$.
\end{restatedresult}

\begin{proof}[Proof of Proposition~\ref{prop:harris_recurrence}]
Fix $\gamma>0$. Since $U$ and $\beta$ are finite-valued, both
$\pi(x)$ and $q_\gamma(x,y)$ are strictly positive and finite for all
$x,y$. It follows that the accepted-move density satisfies
\[
p_\gamma(x,y)=q_\gamma(x,y)\alpha_\gamma(x,y)
=\min\left\{q_\gamma(x,y),
\frac{\pi(y)}{\pi(x)}q_\gamma(y,x)\right\}>0.
\]
Consequently, if $\lambda(A)>0$, then
\begin{equation}
\label{eq:masla_one_step_irreducibility}
R_\gamma(x,A)\geq\int_A p_\gamma(x,y)\,\lambda(dy)>0
\qquad\text{for every }x.
\end{equation}
For completeness, strict positivity of the integral follows by writing
\[
A=\bigcup_{k\geq1}\{y\in A:p_\gamma(x,y)\geq1/k\}.
\]
At least one of these sets has positive Lebesgue measure, and its
contribution to the integral is strictly positive. Thus every set of
positive Lebesgue measure can be reached in one step from every state,
which proves $\lambda$-irreducibility. Also, $\pi$ and $\lambda$ have
the same null sets because $\pi$ has a strictly positive density, so
the chain is $\pi$-irreducible as well. Taking $A=\mathbb{R}^d$ in the
integral above shows that
\begin{equation}
\label{eq:masla_positive_acceptance}
1-r_\gamma(x)=\int_{\mathbb{R}^d}p_\gamma(x,y)\,\lambda(dy)>0
\qquad\text{for every }x.
\end{equation}

We now prove Harris recurrence directly. Fix a Borel set $A$ with
$\lambda(A)>0$, and define its hitting time and avoidance probability by
\[
\tau_A:=\inf\{n\geq0:X_n\in A\},
\qquad h(x):=\mathbb{P}_x(\tau_A=\infty),
\]
where $\inf\varnothing=\infty$. The function $h$ is measurable,
$0\leq h\leq1$, and $h=0$ on $A$. By conditioning on the first step,
\begin{equation}
\label{eq:masla_avoidance_harmonic}
h(x)=(R_\gamma h)(x):=\int h(y)R_\gamma(x,dy),
\qquad x\in A^c.
\end{equation}
Proposition~\ref{thm:stationary} gives invariance of $\pi$. Hence
\begin{align*}
\int_{\mathbb{R}^d}(R_\gamma h)(x)\,\pi(dx)
&=\int_{\mathbb{R}^d}h(x)\,\pi(dx)\\
&=\int_{A^c}h(x)\,\pi(dx)
=\int_{A^c}(R_\gamma h)(x)\,\pi(dx).
\end{align*}
All these integrals are finite since $h\leq1$. Subtracting the
$A^c$ integral gives
\[
\int_A(R_\gamma h)(x)\,\pi(dx)=0.
\]
Since $R_\gamma h\geq0$ and $\pi(A)>0$, there exists $x_*\in A$
with $(R_\gamma h)(x_*)=0$. Expanding the kernel and using $h(x_*)=0$
yields
\[
0=(R_\gamma h)(x_*)
=\int_{\mathbb{R}^d}p_\gamma(x_*,y)h(y)\,\lambda(dy).
\]
The integrand is nonnegative, so it is zero for $\lambda$-almost
every $y$. Since $p_\gamma(x_*,y)>0$ everywhere, this implies
$h(y)=0$ for $\lambda$-almost every $y$.

It remains to obtain the conclusion at every initial point, rather
than only almost everywhere. For any $x\in A^c$, expanding
\eqref{eq:masla_avoidance_harmonic} now gives
\[
h(x)=\int_{\mathbb{R}^d}p_\gamma(x,y)h(y)\,\lambda(dy)
+r_\gamma(x)h(x)=r_\gamma(x)h(x).
\]
The integral vanishes because $h=0$ Lebesgue almost everywhere.
By \eqref{eq:masla_positive_acceptance}, $1-r_\gamma(x)>0$, so
$h(x)=0$. Together with $h=0$ on $A$, this proves
$\mathbb{P}_x(\tau_A<\infty)=1$ for every $x$.

There are in fact infinitely many visits. For each integer $m\geq0$,
the Markov property at time $m$ and $h\equiv0$ give
\[
\mathbb{P}_x(X_n\notin A\text{ for every }n\geq m)
=\mathbb{E}_x[h(X_m)]=0.
\]
The event of only finitely many visits to $A$ is the countable union
of these events over $m$, and therefore has probability zero. This
establishes Harris recurrence. Since $\pi$ is an invariant probability
distribution (Proposition \ref{thm:stationary}), the chain is positive Harris recurrent. The standard
uniqueness theorem for invariant measures of Harris recurrent chains
then implies that $\pi$ is the unique invariant probability
distribution; see Theorem~11.2.5 of \cite{douc2018markov}.
\end{proof}

\begin{theorem}[A geometric drift criterion]
\label{thm:geometric_ergodicity}
Fix $\gamma>0$ and assume the conditions of
Proposition~\ref{thm:stationary}.

Suppose there exist a finite-valued Borel function
$V:\mathbb{R}^d\to[1,\infty)$, a Borel set $C$, constants
$a\in(0,1)$ and $b\in[0,\infty)$, a probability measure $\nu$, and
$\epsilon\in(0,1]$ such that
\begin{align}
R_\gamma V(x)
&:=\int_{\mathbb{R}^d}V(y)R_\gamma(x,dy)
\leq aV(x)+b\mathbbm{1}_C(x),
&&x\in\mathbb{R}^d,
\label{eq:masla_geometric_drift}\\
R_\gamma(x,A)&\geq\epsilon\nu(A),
&&x\in C,\quad A\in\mathcal{B}(\mathbb{R}^d),
\label{eq:masla_geometric_minorization}
\end{align}
and $\nu(C)>0$. Then $\pi(V):=\int V\,d\pi<\infty$, and $R_\gamma$
is $V$-uniformly geometrically ergodic: there exist constants
$M_\gamma<\infty$ and $\rho_\gamma\in(0,1)$ such that
\begin{equation}
\label{eq:masla_geometric_convergence}
\|R_\gamma^n(x,\cdot)-\pi\|_V
\leq M_\gamma V(x)\rho_\gamma^n,
\qquad x\in\mathbb{R}^d,\quad n\geq0.
\end{equation}
Here, for a signed measure $\xi$ with $\int V\,d|\xi|<\infty$, the
weighted total variation norm is
\[
\|\xi\|_V:=\sup_{f\text{ Borel},\,|f|\leq V}
\left|\int f\,d\xi\right|.
\]
In particular, convergence holds geometrically in total variation
from every initial point. More generally, for any initial probability
measure $\mu$ with $\mu(V)<\infty$,
\begin{equation}
\label{eq:masla_geometric_initial_law}
\|\mu R_\gamma^n-\pi\|_V
\leq M_\gamma\mu(V)\rho_\gamma^n,
\qquad n\geq0.
\end{equation}
\end{theorem}

\begin{proof}
By Proposition~\ref{prop:harris_recurrence}, $R_\gamma$ is
$\lambda$-irreducible and positive Harris recurrent, with unique
invariant probability distribution $\pi$. The minorization
\eqref{eq:masla_geometric_minorization} makes $C$ a one-step small
set with minorizing measure $\epsilon\nu$, and every small set is
petite. We verify aperiodicity directly using the everywhere-positive
accepted-move density: every Borel set of
positive $\lambda$-measure is reached in one step, by
\eqref{eq:masla_one_step_irreducibility}. In particular, no cyclic
decomposition of period $k\geq2$ is possible. Such a decomposition
would have disjoint sets $D_0,\ldots,D_{k-1}$ of positive
$\lambda$-measure with transitions from $D_j$ to $D_{j+1}$
(indices modulo $k$) having probability one. But for $x\in D_j$,
\eqref{eq:masla_one_step_irreducibility} gives
$R_\gamma(x,D_j)>0$, contradicting that requirement. Thus
$R_\gamma$ is aperiodic.

We next verify integrability of $V$ without presupposing
$\pi(V)<\infty$. Dropping the indicator from
\eqref{eq:masla_geometric_drift} and iterating gives
\[
R_\gamma^n V(x)\leq a^n V(x)+b\sum_{j=0}^{n-1}a^j
\leq a^nV(x)+\frac{b}{1-a}.
\]
For $L>0$, invariance of $\pi$ and the inequalities
$R_\gamma^n(V\wedge L)\leq R_\gamma^n V$ and
$R_\gamma^n(V\wedge L)\leq L$ imply
\[
\pi(V\wedge L)
=\pi\bigl(R_\gamma^n(V\wedge L)\bigr)
\leq\int\left[\left(a^nV(x)+\frac{b}{1-a}\right)\wedge L\right]
\pi(dx).
\]
Since $V(x)<\infty$ for every $x$ and the integrand is bounded by
$L$, dominated convergence as $n\to\infty$ gives
$\pi(V\wedge L)\leq [b/(1-a)]\wedge L$. Monotone convergence as
$L\to\infty$ then yields $\pi(V)\leq b/(1-a)<\infty$.
Integrating the original drift inequality is now legitimate and gives
the sharper bound
\[
\pi(V)=\pi(R_\gamma V)\leq a\pi(V)+b\pi(C),
\qquad \pi(V)\leq\frac{b\pi(C)}{1-a}.
\]

The geometric drift theorem, Theorem~15.2.4(a) of
\cite{douc2018markov}, applies to an irreducible, aperiodic Markov
kernel satisfying a geometric drift inequality toward a petite set.
All of its hypotheses have now been verified for $R_\gamma$, $V$,
and $C$. It gives constants $M_\gamma<\infty$ and
$\rho_\gamma\in(0,1)$ satisfying
\eqref{eq:masla_geometric_convergence}.

For a probability measure $\mu$ with $\mu(V)<\infty$, the iterated
drift bound gives $(\mu R_\gamma^n)(V)<\infty$. For every Borel
function $f$ with $|f|\leq V$, we may therefore write
\begin{align*}
|\mu R_\gamma^n f-\pi(f)|
&=\left|\int\bigl(R_\gamma^n f(x)-\pi(f)\bigr)\,\mu(dx)\right|\\
&\leq\int M_\gamma V(x)\rho_\gamma^n\,\mu(dx)
=M_\gamma\mu(V)\rho_\gamma^n.
\end{align*}
Taking the supremum over these functions proves
\eqref{eq:masla_geometric_initial_law}. Finally, $V\geq1$ implies
$\|\xi\|_{\mathrm{TV}}\leq\|\xi\|_V$, where
$\|\xi\|_{\mathrm{TV}}:=\sup_{|f|\leq1}|\int f\,d\xi|$, so the
same bounds imply geometric convergence in total variation.
\end{proof}

\begin{remark}[Verifying the drift and minorization]
The geometric ergodicity conclusion requires
\eqref{eq:masla_geometric_drift}; exact invariance and Harris
recurrence alone do not establish a geometric rate. The drift must
be checked for the full MASLA kernel, including rejection. In terms
of the accepted-move density, its exact expression is
\begin{align*}
R_\gamma V(x)
&=\int V(y)p_\gamma(x,y)\,\lambda(dy)+r_\gamma(x)V(x)\\
&=V(x)+\int\bigl(V(y)-V(x)\bigr)p_\gamma(x,y)\,\lambda(dy).
\end{align*}
Thus, for a given potential and selection, it suffices to find
$V\geq1$ and $a<1$ such that the last integral is at most
$-(1-a)V(x)$ for $x\notin C$, and to bound
$\sup_{x\in C}R_\gamma V(x)<\infty$. A drift estimate for the
unadjusted Gaussian proposal by itself does not establish this
inequality for the accepted chain.

If $C$ is compact with $\lambda(C)>0$, and $U$ and $\beta$ are
bounded on compact sets, Proposition~\ref{prop:Aperiodicity} below
provides \eqref{eq:masla_geometric_minorization} with
$\nu=\lambda(\cdot\cap C)/\lambda(C)$, so that $\nu(C)=1$.
The remaining task is then to verify the drift inequality for the
particular target and step size. The function $V$, the drift and
minorization constants, and the resulting $M_\gamma,\rho_\gamma$
may depend on $\gamma$; no uniform rate over all positive step sizes
is asserted.
\end{remark}

\begin{proposition}[Small sets and strong aperiodicity]
	\label{prop:Aperiodicity}
	Fix $\gamma>0$ and assume the conditions of
	Proposition~\ref{thm:stationary}. Suppose also that $U$ and $\beta$
	are bounded on compact sets. Then every compact set
	$C\subset\mathbb{R}^d$ with $\lambda(C)>0$ is an accessible one-step
	small set for $R_\gamma$. More precisely, define
	\[
	\nu_C(A):=\frac{\lambda(A\cap C)}{\lambda(C)},
	\qquad A\in\mathcal{B}(\mathbb{R}^d),
	\]
	and let
	\begin{align*}
		L_C&:=\sup_{x\in C}\|\beta(x)\|,
		&\Delta_C&:=\sup_{x\in C}U(x)-\inf_{x\in C}U(x),\\
		D_C&:=\sup_{x,y\in C}\|x-y\|.
	\end{align*}
	Then the constant
	\begin{equation}
		\label{eq:masla_compact_minorization_constant}
		\epsilon_{\gamma,C}
		:=\lambda(C)(4\pi\gamma)^{-d/2}
		\exp\left(-\frac{(D_C+\gamma L_C)^2}{4\gamma}-\Delta_C\right)
	\end{equation}
	belongs to $(0,1]$ and satisfies
	\begin{equation}
		\label{eq:masla_compact_minorization}
		R_\gamma(x,A)\geq\epsilon_{\gamma,C}\nu_C(A),
		\qquad x\in C,\quad A\in\mathcal{B}(\mathbb{R}^d).
	\end{equation}
	In particular, $R_\gamma$ is strongly aperiodic, as well as
	$\lambda$-irreducible and $\pi$-irreducible.
\end{proposition}

\begin{proof}
	Fix a compact set $C$ with $\lambda(C)>0$. Compactness gives
	$\lambda(C)<\infty$ and $D_C<\infty$, while the boundedness
	assumptions give $L_C<\infty$ and $0\leq\Delta_C<\infty$.
	Thus $\nu_C$ is a probability measure with $\nu_C(C)=1$.
	
	For $x,y\in C$, the triangle inequality yields
	\[
	\|y-x+\gamma\beta(x)\|\leq D_C+\gamma L_C,
	\qquad
	\|x-y+\gamma\beta(y)\|\leq D_C+\gamma L_C.
	\]
	By \eqref{eq:masla_proposal_density}, both $q_\gamma(x,y)$ and
	$q_\gamma(y,x)$ are therefore bounded below by
	\[
	m_{\gamma,C}:=(4\pi\gamma)^{-d/2}
	\exp\left(-\frac{(D_C+\gamma L_C)^2}{4\gamma}\right)>0.
	\]
	Also, for $x,y\in C$,
	\[
	\frac{\pi(y)}{\pi(x)}=e^{U(x)-U(y)}\geq e^{-\Delta_C}.
	\]
	It follows that the accepted-move density satisfies
	\begin{align*}
		p_\gamma(x,y)
		&:=q_\gamma(x,y)\alpha_\gamma(x,y)\\
		&=\min\left\{q_\gamma(x,y),
		\frac{\pi(y)}{\pi(x)}q_\gamma(y,x)\right\}\\
		&\geq\min\{m_{\gamma,C},m_{\gamma,C}e^{-\Delta_C}\}
		=m_{\gamma,C}e^{-\Delta_C},
		\qquad x,y\in C.
	\end{align*}
	Here the last equality uses $\Delta_C\geq0$.
	
	For every $x\in C$ and every Borel set $A$, the kernel formula
	\eqref{eq:masla_kernel} and nonnegativity of the rejection term now give
	\begin{align*}
		R_\gamma(x,A)
		&=\int_A p_\gamma(x,y)\,\lambda(dy)
		+r_\gamma(x)\mathbbm{1}_A(x)\\
		&\geq\int_{A\cap C}p_\gamma(x,y)\,\lambda(dy)\\
		&\geq m_{\gamma,C}e^{-\Delta_C}\lambda(A\cap C)
		=\epsilon_{\gamma,C}\nu_C(A).
	\end{align*}
	The constant $\epsilon_{\gamma,C}$ is strictly positive because all
	its factors are strictly positive. Taking $A=C$ in this inequality
	also shows that $\epsilon_{\gamma,C}\leq R_\gamma(x,C)\leq1$.
	This proves the asserted one-step minorization and hence the small-set
	property (Definition~9.1.1 of \cite{douc2018markov}).
	
	By Proposition~\ref{prop:harris_recurrence}, the kernel is
	$\lambda$-irreducible and $\pi$-irreducible. Moreover, its one-step
	positivity property \eqref{eq:masla_one_step_irreducibility} gives
	\[
	R_\gamma(x,C)>0,\qquad x\in\mathbb{R}^d,
	\]
	since $\lambda(C)>0$. Thus $C$ is accessible from every state.
	Finally, the minorizing measure $\epsilon_{\gamma,C}\nu_C$ assigns
	mass $\epsilon_{\gamma,C}\nu_C(C)=\epsilon_{\gamma,C}>0$ to $C$.
	The existence of this accessible one-step small set proves strong
	aperiodicity, by Definitions~9.1.2 and~9.3.5 of
	\cite{douc2018markov}.
\end{proof}

Theorem~\ref{thm:masla_geometric_ergodicity} verifies the drift condition for a class of nonsmooth
potentials, allowing nonconvex perturbations and discontinuous
selections. It establishes the hypotheses of
Theorem~\ref{thm:geometric_ergodicity} from conditions directly on
the potential and its conservative field.

\begin{restatedresult}{Theorem~\ref{thm:masla_geometric_ergodicity} (MASLA geometric ergodicity)}
Let $H\in\mathbb{R}^{d\times d}$ be symmetric and positive definite.
Write $h_-$ and $h_+$ for its smallest and largest eigenvalues.
Suppose that
\[
U(x)=\tfrac12 x^\top Hx+g(x),
\qquad |g(x)-g(y)|\leq L\|x-y\|,
\qquad x,y\in\mathbb{R}^d,
\]
for some $L<\infty$. Assume that $g$ admits a conservative field
$D_g$ whose values are uniformly bounded: for some $B<\infty$,
\[
\sup_{x\in\mathbb{R}^d}\sup_{v\in D_g(x)}\|v\|\leq B.
\]
Choose any fixed Borel measurable selection $v(x)\in D_g(x)$ and set
\[
D_U(x):=Hx+D_g(x),\qquad \beta(x):=Hx+v(x).
\]
For every step size
\begin{equation*}
0<\gamma<\frac{2}{h_+},
\end{equation*}
the MASLA kernel using this selection and the exact MH acceptance
probability is $V$-uniformly geometrically ergodic with
$V(x)=1+\|x\|^2$. Specifically, $\pi(dx)\propto e^{-U(x)}\lambda(dx)$
is well-defined, and there exist $M_\gamma<\infty$ and
$\rho_\gamma\in(0,1)$ such that
\[
\|R_\gamma^n(x,\cdot)-\pi\|_V
\leq M_\gamma(1+\|x\|^2)\rho_\gamma^n,
\qquad x\in\mathbb{R}^d,\quad n\geq0.
\]
This includes initial points at which $U$ is not differentiable.
More precisely, writing $a_\gamma:=\|I_d-\gamma H\|_{\mathrm{op}}<1$,
the rejection probability and the drift satisfy
\begin{equation*}
\lim_{\|x\|\to\infty}r_\gamma(x)=0,
\qquad
\limsup_{\|x\|\to\infty}\frac{R_\gamma V(x)}{V(x)}
\leq a_\gamma^2<1.
\end{equation*}
\end{restatedresult}

\begin{proof}
The Lipschitz bound on $g$ gives
\[
U(x)\geq\tfrac12 h_-\|x\|^2-L\|x\|+g(0).
\]
Consequently, $e^{-U}$ is integrable and its normalizing constant is
strictly positive. The chain rule for $D_g$, together with the usual
chain rule for the quadratic term, shows that $D_U$ is a conservative
field for $U$. The selection $\beta$ is Borel measurable and bounded
on compact sets, as is $U$. Proposition~\ref{thm:stationary} therefore
applies. Fix a step size satisfying \eqref{eq:masla_quadratic_stepsize}
and abbreviate $a=a_\gamma$. The eigenvalues of $I_d-\gamma H$ lie
in $(-1,1)$, so $a<1$.

We first control the rejection probability without differentiating
$g$. Define the log MH ratio by
\[
\ell_\gamma(x,y):=
\log\frac{\pi(y)q_\gamma(y,x)}{\pi(x)q_\gamma(x,y)}.
\]
Substituting $U(x)=x^\top Hx/2+g(x)$ and $\beta(x)=Hx+v(x)$ into
the Gaussian proposal formula and expanding the squared norms yields
\begin{equation}
\label{eq:masla_quadratic_log_ratio}
\ell_\gamma(x,y)
=\frac{\gamma}{4}\bigl(\|Hx\|^2-\|Hy\|^2\bigr)
+E_\gamma(x,y),
\end{equation}
where
\begin{align*}
E_\gamma(x,y)
&=g(x)-g(y)+\tfrac12\langle y-x,v(x)+v(y)\rangle-\tfrac{\gamma}{2}
\bigl(\langle Hy,v(y)\rangle-\langle Hx,v(x)\rangle\bigr)\\
&\quad-\tfrac{\gamma}{4}\bigl(\|v(y)\|^2-\|v(x)\|^2\bigr).
\end{align*}
In obtaining \eqref{eq:masla_quadratic_log_ratio}, the terms
$\tfrac12(x^\top Hx-y^\top Hy)$ cancel with the corresponding
cross terms from the two proposal densities. Since
$\|v(x)\|,\|v(y)\|\leq B$ and $\|H\|_{\mathrm{op}}=h_+$,
\begin{align*}
|E_\gamma(x,y)|
&\leq(L+B)\|x-y\|
+\tfrac{\gamma h_+B}{2}(\|x\|+\|y\|)
+\tfrac{\gamma B^2}{4}\\
&\leq K_\gamma(\|x\|+\|y\|)+\tfrac{\gamma B^2}{4},
\end{align*}
where $K_\gamma:=L+B+\tfrac{\gamma h_+B}{2}$. 
These bounds hold at all points, including nonsmooth points, and use
only the boundedness of the selected field values at both $x$ and $y$.

Write a proposal from $x$ as
\[
Y=(I_d-\gamma H)x+w_x(Z),
\qquad w_x(Z):=-\gamma v(x)+\sqrt{2\gamma}\,Z,
\qquad Z\sim N(0,I_d).
\]
Fix $T>0$. On the event $\|Z\|\leq T$, we have
$\|w_x(Z)\|\leq t_T:=\gamma B+\sqrt{2\gamma}\,T$ and
$\|Y\|\leq a\|x\|+t_T$. Since $H$ commutes with $I_d-\gamma H$,
\[
\|H(I_d-\gamma H)x\|\leq a\|Hx\|.
\]
Expanding $\|HY\|^2$ and applying the triangle and Cauchy--Schwarz
inequalities gives, on the same event,
\begin{align*}
\|Hx\|^2-\|HY\|^2
&\geq(1-a^2)\|Hx\|^2
-2a\|Hx\|h_+t_T-h_+^2t_T^2\\
&\geq(1-a^2)h_-^2\|x\|^2
-2ah_+^2t_T\|x\|-h_+^2t_T^2.
\end{align*}
Combining this bound with \eqref{eq:masla_quadratic_log_ratio} yields
\begin{align*}
\ell_\gamma(x,Y)
&\geq\frac{\gamma}{4}(1-a^2)h_-^2\|x\|^2-\left(\frac{\gamma}{2}ah_+^2t_T
+K_\gamma(1+a)\right)\|x\|\\
&\quad-\frac{\gamma}{4}h_+^2t_T^2-K_\gamma t_T
-\frac{\gamma B^2}{4}.
\end{align*}
The coefficient of $\|x\|^2$ is strictly positive. Hence there is
$R_T<\infty$ such that $\ell_\gamma(x,Y)\geq0$ whenever
$\|x\|\geq R_T$ and $\|Z\|\leq T$. Such proposals are accepted
with probability one. Therefore,
\[
0\leq r_\gamma(x)
=\mathbb{E}\bigl[1-\alpha_\gamma(x,Y)\bigr]
\leq\mathbb{P}(\|Z\|>T),
\qquad \|x\|\geq R_T.
\]
First taking the limit superior as $\|x\|\to\infty$, and then
letting $T\to\infty$, proves $r_\gamma(x)\to0$.

We now derive the drift inequality for the accepted chain. For
$V(x)=1+\|x\|^2$, the Gaussian second-moment formula gives
\begin{align*}
\mathbb{E}[V(Y)]
&=1+\|(I_d-\gamma H)x-\gamma v(x)\|^2+2\gamma d\\
&\leq1+a^2\|x\|^2+2a\gamma B\|x\|+\gamma^2B^2+2\gamma d.
\end{align*}
Keeping the rejection contribution explicitly and using $V\geq0$,
\begin{align*}
R_\gamma V(x)
&=\mathbb{E}\bigl[\alpha_\gamma(x,Y)V(Y)
+(1-\alpha_\gamma(x,Y))V(x)\bigr]\\
&\leq\mathbb{E}[V(Y)]+r_\gamma(x)V(x).
\end{align*}
Dividing by $V(x)$ and using $r_\gamma(x)\to0$ proves the second
assertion of \eqref{eq:masla_quadratic_tail_bounds}.

Choose any $\theta\in(a^2,1)$. There exists $R\geq1$ such that
$R_\gamma V(x)\leq\theta V(x)$ for $\|x\|>R$.
For the compact ball $C=\{x:\|x\|\leq R\}$, the preceding moment
bound and $r_\gamma\leq1$ give
\[
\sup_{x\in C}R_\gamma V(x)
\leq 2+(1+a^2)R^2+2a\gamma BR+\gamma^2B^2+2\gamma d=:b_C<\infty.
\]
Consequently, the full MASLA kernel satisfies the derived drift
inequality
\[
R_\gamma V(x)\leq\theta V(x)+b_C\mathbbm{1}_C(x),
\qquad x\in\mathbb{R}^d.
\]
Since $U$ and $\beta$ are bounded on compact sets,
Proposition~\ref{prop:Aperiodicity} supplies the minorization on $C$
with $\nu_C(C)=1$. Theorem~\ref{thm:geometric_ergodicity} now gives
the stated $V$-uniform geometric convergence. Every estimate above
holds for the chosen selection at every point; no differentiability
of $g$ or continuity of $v$ was required.
\end{proof}

\begin{example}[A nonsmooth nonconvex target]
\label{ex:masla_nonconvex_geometric}
Fix $\tau>0$ and a symmetric positive definite matrix $H$. Set
\[
U(x)=\tfrac12 x^\top Hx-\tau\|x\|_1.
\]
An additive constant may be included if a nonnegative potential is
desired; it does not change the target distribution or the algorithm.
The perturbation $g(x)=-\tau\|x\|_1$ is  Lipschitz with
constant $\tau\sqrt d$. A conservative field for $g$ is
\[
D_g(x)=\prod_{j=1}^d G_j(x_j),
\qquad
G_j(t)=
\begin{cases}
\{-\tau\},&t>0,\\
[-\tau,\tau],&t=0,\\
\{\tau\},&t<0.
\end{cases}
\]
Indeed, along an absolutely continuous curve each coordinate has
derivative zero almost everywhere on its zero set, while the usual
derivative of $-\tau|t|$ applies off that set. Thus every element of
$D_g$ satisfies the chain rule, and the field has closed graph and
nonempty compact values. Moreover, $\|v\|\leq\tau\sqrt d$ for every
$v\in D_g(x)$.

The potential is nonsmooth on the coordinate hyperplanes and is
nonconvex: for a coordinate vector $e_j$ and sufficiently small
$s>0$,
\[
\frac{U(se_j)+U(-se_j)}2
=\tfrac12 H_{jj}s^2-\tau s<U(0)=0.
\]
Nevertheless, Theorem~\ref{thm:masla_geometric_ergodicity} proves
geometric ergodicity for every $0<\gamma<2/h_+$ and every Borel
selection from $D_g$, including arbitrary choices in $[-\tau,\tau]$
when a coordinate is zero. This verifies geometric ergodicity for
a concrete nonsmooth, nonconvex MASLA target without assuming a
kernel drift inequality.
\end{example}

\clearpage
\section{MASLA implementation and automatic differentiation}
\label{app:implementation}

\subsection{Explicit MASLA pseudocode}
\label{sec:masla_pseudocode}
Algorithm~\ref{alg:masla_explicit} implements the kernel in Proposition~\ref{thm:stationary}. The input potential is evaluated up to an additive constant, so its normalizing constant is not required. The selection $\beta$ is fixed before simulation and used at both the current state and the proposal.

\begin{algorithm}[MASLA with a fixed measurable selection]
\label{alg:masla_explicit}
\normalfont
\textbf{Input:} A finite-valued potential $U$ satisfying \eqref{eq:standing_target}; a finite-valued Borel selection $\beta(x)\in D_U(x)$; a step size $\gamma>0$; an initial point $x_0\in\mathbb R^d$; a number of updates $N$; and a burn-in length $0\leq B<N$.
\begin{enumerate}
\item Initialize $x\leftarrow x_0$, $u\leftarrow U(x)$, $b\leftarrow\beta(x)$, and the acceptance count $A\leftarrow0$.
\item For $n=0,\ldots,N-1$:
\begin{enumerate}[(a)]
\item Draw $z\sim N(0,I_d)$ and $a\sim\operatorname{Uniform}(0,1)$, independently of each other and of previous draws.
\item Propose $y\leftarrow x-\gamma b+\sqrt{2\gamma}\,z$.
\item Evaluate $u_y\leftarrow U(y)$ and $b_y\leftarrow\beta(y)$ using the same selection rule as at $x$.
\item Compute the log MH ratio
\[
\ell\leftarrow u-u_y
-\frac{\|x-y+\gamma b_y\|^2}{4\gamma}
+\frac{\|y-x+\gamma b\|^2}{4\gamma}.
\]
\item If $\log a<\min\{0,\ell\}$, set
$(x,u,b)\leftarrow(y,u_y,b_y)$ and $A\leftarrow A+1$.
Otherwise, retain $(x,u,b)$.
\item Record $X_{n+1}\leftarrow x$, including repeated states following rejection.
\end{enumerate}
\item Return $(X_{B+1},\ldots,X_N)$ and the acceptance fraction $A/N$.
\end{enumerate}
\end{algorithm}

The two squared-norm terms are respectively the reverse and forward Gaussian log proposal densities, with their common constants canceled. In particular, the reverse density requires $\beta(y)$ even when the proposal is subsequently rejected. Caching $U(x)$ and $\beta(x)$ means that each update requires only one new potential evaluation and one new selection evaluation, at $y$. The logarithmic comparison avoids exponentiating a large log ratio. A burn-in choice alone does not certify convergence.

\subsection{Choosing the selection with automatic differentiation}
\label{sec:autodiff_selection}
Suppose $U$ is represented by a fixed finite computational graph of locally Lipschitz primitives, each equipped with a locally bounded conservative field. Reverse-mode automatic differentiation, using a fixed Borel selection from each primitive field, then produces a selection $\beta_{\mathrm{AD}}$ from a conservative field for $U$; see \citet[Section~5, Theorem~8]{bolte2021conservative}. Smooth primitives use their usual derivatives. For absolute value, choose $\operatorname{sign}(t)$ for $t\neq0$ and a fixed value in $[-1,1]$ at zero; for ReLU, choose $0$ below zero, $1$ above zero, and a fixed value in $[0,1]$ at zero. Fixing these conventions and the computational graph makes the selection reproducible. This conclusion requires valid primitive rules: it is not a blanket guarantee for arbitrary custom backward operations or branching programs.

The resulting conservative field can be larger than the Clarke subdifferential of the composite potential. Thus, one should not automatically identify $\beta_{\mathrm{AD}}(x)$ with an element of $\partial U(x)$ at every point. For the conservative-field formulation of MASLA, use a field containing the computed selection and verify any additional assumptions of the convergence result being invoked.

For the quartic target in Appendix~\ref{sec:high_dimensional_quartic}, differentiating the expression $\sum_i(x_i^4/4-x_i^2+|x_i|)$ with the absolute-value derivative set to zero at zero gives exactly
\[
\beta_{\mathrm{AD},i}(x)=x_i^3-2x_i+\operatorname{sign}(x_i),
\qquad \operatorname{sign}(0)=0,
\]
which belongs to the explicit Clarke field in \eqref{eq:high_dimensional_quartic_selection}. This provides a direct check of the autodifferentiation output, including on the nonsmooth hyperplanes. In practice, inspect the library's primitive rules and specify a custom rule when needed to enforce the chosen convention.

For Algorithm~\ref{alg:masla_explicit}, all evaluations must use the same deterministic map $\beta$: independently resampling a subgradient or a minibatch at each evaluation does not implement its stated Gaussian proposal kernel. The MH ratio must also use the full potential $U$. Proposition~\ref{thm:stationary} gives exact invariance for any fixed finite-valued Borel drift when its actual proposal density is used, but interpreting the drift as a conservative selection and applying the later convergence results require their respective hypotheses. In particular, automatic differentiation by itself does not establish the pointwise gradient agreement required by Proposition~\ref{prop:selection_independence}.

\clearpage
	\section{Additional experimental details}\label{app:experiments}
	\label{sec:Numerical}
	
	This appendix details the convex composite benchmark, competing algorithms, and distance comparisons in Section~\ref{sec:Comparative}, then examines the nonsmooth nonconvex example, its proximal nonuniqueness and explicit Clarke subgradient, and the effect of MH correction in Section~\ref{sec:nonapplicable}. Section~\ref{sec:high_dimensional_quartic} adds a $100$-dimensional nonsmooth nonconvex quartic experiment examining numerical escape of USLA and the effect of the MH correction.
	
	\subsection{MASLA Versus State-of-the-Art Algorithms}
	\label{sec:Comparative}
	
	\subsubsection{Target Distribution Description}
	To ensure a fair comparison, we adopt the exact experimental setup in \cite{habring2024subgradient} on sampling from a two-dimensional density with the potential 
	$$U(x) = F(x) + G(Kx),$$ 
	where $F: \mathbb{R}^d \rightarrow [0, \infty)$ and $G: \mathbb{R}^{d'} \rightarrow [0, \infty)$ are both proper, convex, and lower semi-continuous, and $K: \mathbb{R}^d \rightarrow \mathbb{R}^{d'}$ is a linear operator. They considered function $F$ as either Lipschitz continuous with constant $L_F$ or differentiable with a Lipschitz continuous gradient with constant $L_{\nabla F}$, while $G$ as convex and Lipschitz continuous with constant $L_G$.  They used the total variation (TV) functional for $G(Kx)$, defined in the two-dimensional case as $G: \mathbb{R} \to \mathbb{R}$ and $K: \mathbb{R}^2 \to \mathbb{R}$ by
	\[
	G(p) = \lambda |p|, \qquad Kx = x_2 - x_1,
	\]
	where $\lambda > 0$ is a scaling parameter.  It is straightforward to verify that $G$ is Lipschitz continuous with Lipschitz constant $L_G = \lambda$, and its subdifferential is given by
	\[
	\partial G(p) = 
	\begin{cases}
		[-\lambda, \lambda] & \text{if } p = 0, \\
		\{\text{sign}(p)\lambda\} & \text{otherwise}.
	\end{cases}
	\]
	The operator $K$ satisfies $\|K\|^2 \leq 2$. For $F: \mathbb{R}^2 \to \mathbb{R}$, they used the squared $L^2$ discrepancy, defined as
	\[
	F(x) = \frac{1}{2\sigma^2} \|x - y\|_2^2,
	\]
	for $y$ being a known vector,
	which is $m$-strongly convex with $m = \sigma^{-2}$ and has a Lipschitz continuous gradient with constant $L_{\nabla F} = \sigma^{-2}$. Following \cite{habring2024subgradient}, we set $y = (-1, 1)$, $\sigma = 1$, and $\lambda = 5$.

	To approximate the distributions of iterates of the algorithms, we need i.i.d. samples of each iterate. Thus, for each algorithm in this section we compute $10,000$ independent Markov chains and then estimate the distribution of the $k$th iterate as a histogram based on these samples. The initial values of the chains are taken as $(-1, 1)$. Figure \ref{fig:true_MASLA} reports the comparison of the true density and MASLA-sampled distribution after the first step, which is the initial sample distribution, followed by the comparison of the true density and MASLA-sampled distribution after 49,500 iterations. All experiments in this section has no burn-in discartion We can clearly see that the distribution of the first dimension is clustered around value -1 at the beginning, which is expected. After running the MASLA chain, the samples follow the true distribution.

	\begin{figure*}[t!]
		\centering
		{\includegraphics[width = 2.5in]{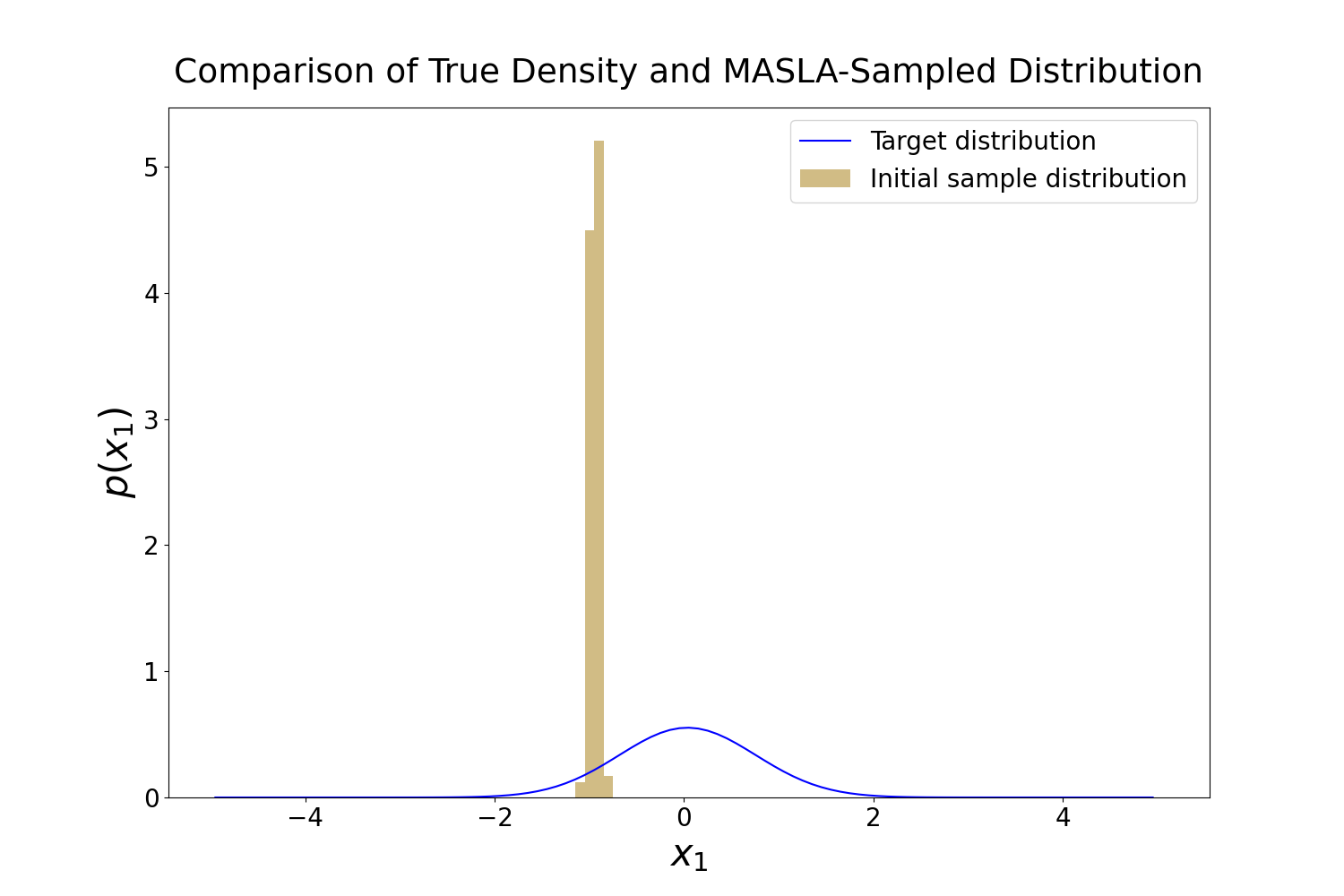}}
		{\includegraphics[width =2.5in]{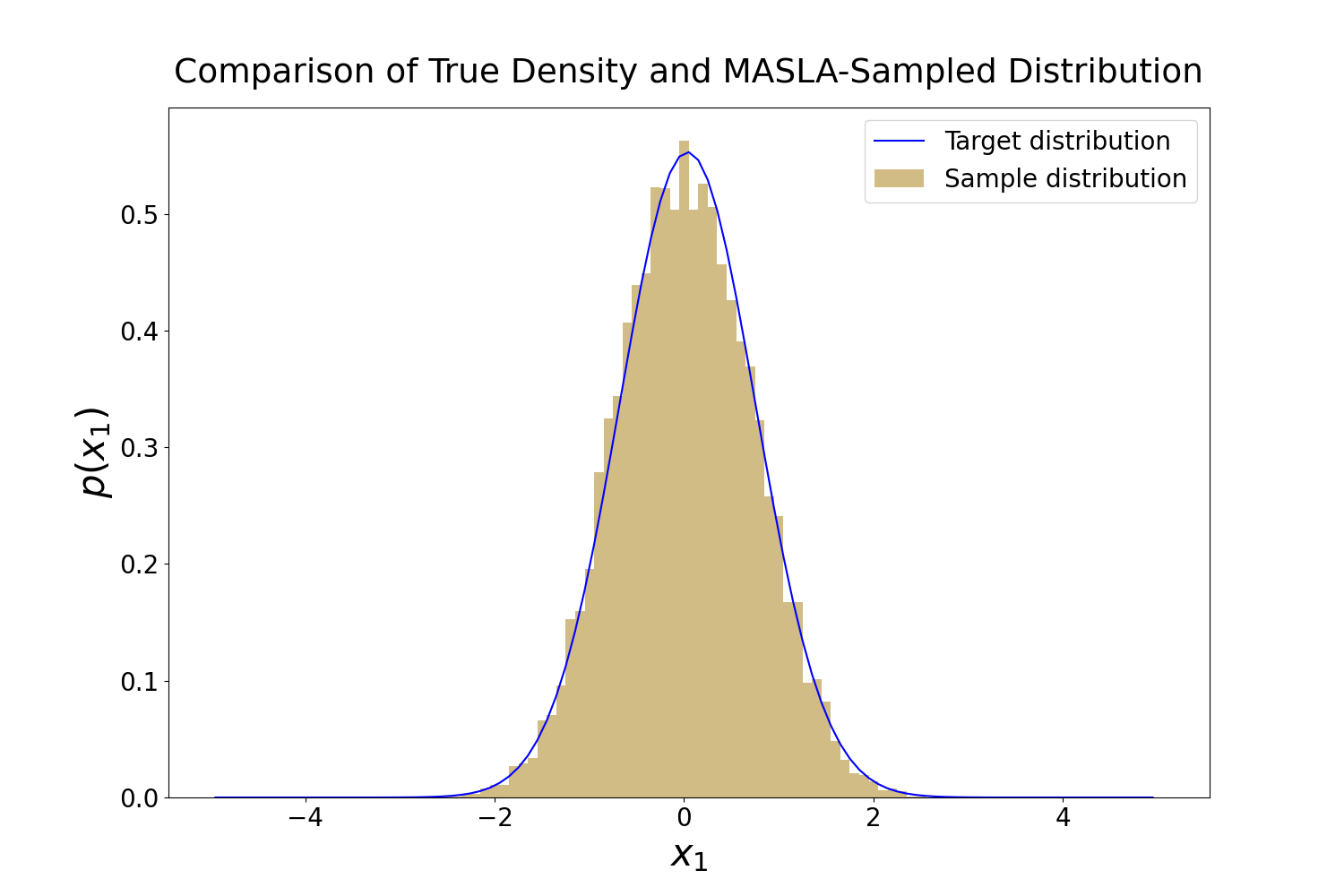}}
		\caption{Comparison of MASLA-sampled initial and final distributions with the true density. The top plot shows the initial sample distribution of the first dimension after the first step, clustered around -1 as expected. The bottom plot displays the distribution after 49,500 iterations, demonstrating convergence to the true density. }
		\label{fig:true_MASLA}
	\end{figure*}
	
	\subsubsection{Comparison in Wasserstein Distance}
\label{sec:expected_wasserstein_rates}
	\cite{habring2024subgradient} proposed two algorithms: the Proximal-subgradient Langevin algorithm (Prox-sub) and the Gradient-subgradient Langevin algorithm (Grad-sub). These methods, building on \cite{durmus2019analysis}, are designed for non-smooth potentials while achieving improved convergence for smooth cases.  In Figure 1 of \cite{habring2024subgradient}, the authors investigated the convergence of Grad-sub and Prox-sub in Wasserstein-2 distance, often shortened to Wasserstein distance. It is a metric on the space of probability measures with finite second moments. Rooted in optimal transport theory, it captures the minimal cost of transporting mass to transform one distribution into another, where cost is measured by squared Euclidean distance. Formally, let $\mu$ and $\nu$ be two probability measures on $\mathbb{R}^d$ with finite second moments, i.e.,
	\[
	\int_{\mathbb{R}^d} \|x\|^2 \, d\mu(x) < \infty \quad \text{and} \quad \int_{\mathbb{R}^d} \|y\|^2 \, d\nu(y) < \infty.
	\]
	The Wasserstein distance between $\mu$ and $\nu$ is defined as
	\[
	W_2(\mu, \nu) := \left( \inf_{\zeta \in \Pi(\mu, \nu)} \int_{\mathbb{R}^d \times \mathbb{R}^d} \|x - y\|^2 \, d\zeta(x, y) \right)^{1/2},
	\]
	where $\Pi(\mu, \nu)$ denotes the set of all couplings of $\mu$ and $\nu$, i.e., all probability measures $\zeta$ on $\mathbb{R}^d \times \mathbb{R}^d$ such that
	\[
	\zeta(A \times \mathbb{R}^d) = \mu(A) \quad \text{and} \quad \zeta(\mathbb{R}^d \times B) = \nu(B)
	\]
	for all Borel sets $A, B \subset \mathbb{R}^d$.
	
	Prox-sub employs a proximal step for the smooth component $F$ and a subgradient step for the nonsmooth $G(Kx)$, ensuring effective handling of nondifferentiability. Let $\partial G(K X_n)$ denote a subgradient of $G$ at $K X_n$, and the update step is given by plugging $Y_{n+1}\in \partial G(K X_n)$ into
	\begin{equation}
		\tag{Prox-sub}
		X_{n+1} = \prox_{\tau F}(X_n - \tau K^* Y_{n+1}) + \sqrt{2\tau} Z_{n+1},
	\end{equation}
	where $\prox_{\tau F}$ is the proximal mapping of $F$ with step size $\tau$ and $K^*$ is the adjoint of $K$. Grad-sub employs an explicit gradient step for the smooth component $F$ and a subgradient step for the nonsmooth $G(Kx)$. The update step is given by plugging $Y_{n+1}\in \partial G(K X_n)$ into:
	\begin{align*}
		\tag{Grad-sub}
		X_{k+1/2} &= X_n - \tau K^*Y_{n+1} \\
		X_{n+1} &= X_{k+1/2} - \tau \nabla F(X_n) + \sqrt{2\tau } Z_{n+1},
	\end{align*}
	where $\nabla F(X_n)$ is the gradient of $F$ at $X_n$.  The formulation of both algorithms ensures nonasymptotic convergence in the Wasserstein-2 distance under appropriate regularity conditions.
	
	Comparison of MASLA, Grad-sub, and Prox-sub for $\TV-L^2$ Sampling in Wasserstein Distance is reported in Figure \ref{fig:Comparison_W2}. The Wasserstein distance between approximate distributions is computed using Python's POT package \citep{flamary2021pot}. Experiments were conducted with step sizes $\tau \in \{10^{-3}, 10^{-4}, 10^{-5}\}$. The source code for Grad-sub and Prox-sub was taken from \cite{habring2024subgradientsource}. Figure \ref{fig:Comparison_W2} shows that numerical convergence rates exceed theoretical expectations for all $\tau$ values. Experimental results indicate negligible differences in convergence rates among Prox-sub, Grad-sub, and MASLA, consistent with Figure 1 of \cite{habring2024subgradient}. All algorithms exhibit exponential ergodicity to the target density.

\paragraph{Calculation of the expected-rate curves.}
The dashed ``expected rate'' curves in Figure~1 of \citet{habring2024subgradient} evaluate upper bounds on the squared Wasserstein error, rather than fitted empirical rates. Their Theorems~4.13 and~5.2 give, for constant step size $\tau$,
\begin{align*}
B_{\mathrm{Prox}}(k)&=(1-m\tau/2)^k E_0+2C\tau/m,\\
B_{\mathrm{Grad}}(k)&=(1-m\tau)^k E_0+C\tau/m,
\end{align*}
where $E_0=W_2^2(\mu S_\tau^G,\pi)$ is the initial squared Wasserstein error after the subgradient step, $S_\tau^G$ denotes the deterministic update $x\mapsto x-\tau K^*s(Kx)$ with the chosen $s(Kx)\in\partial G(Kx)$, and $C=2dL_{\nabla F}+L_G^2\|K\|^2$. The plotted ordinates are $\log_{10}B(k)$. Each bound combines a geometrically decreasing initial-error term with a discretization-error bound proportional to $\tau$, explaining the eventual plateau of the dashed curves. In the authors' plotting notebook \citep{habring2024subgradientsource}, the initial error is approximated numerically by optimal transport between the stored initial samples and a discretized target. The notebook uses $d=2$, $m=L_{\nabla F}=1$, $L_G=5\sqrt{2}$, and an operator-norm bound of $2$, giving $C=204$. Thus its curves are $(1-\tau)^k\widehat E_0+204\tau$ for Grad-sub and $(1-\tau/2)^k\widehat E_0+408\tau$ for Prox-sub. At $\tau=10^{-3}$, their limiting logarithmic levels are approximately $-0.690$ and $-0.389$, respectively. These constants are conservative: the stated two-dimensional model permits $L_G=5$ and $\|K\|=\sqrt{2}$, giving the tighter value $C=54$. Empirical errors below these bounds are therefore consistent with the theory. The bounds concern Grad-sub and Prox-sub; they do not provide a MASLA convergence rate or establish exact convergence to $\pi$ at a fixed positive step size for the unadjusted methods.

	\subsubsection{Comparison in Total Variation Distance}
	In Figure 3 of \cite{habring2024subgradient}, the authors compared the results of Prox-sub with those obtained using the Proximal Metropolis-Adjusted Langevin Algorithm (P-MALA) by \cite{pereyra2016proximal} and the Moreau-Yosida Unadjusted Langevin Algorithm (MYULA) by  \cite{durmus2022proximal} in TV distance. The TV distance is a fundamental measure of dissimilarity between two probability distributions. Intuitively, it quantifies the maximum difference in probability that the two measures assign to the same event.
	Let $\mu$ and $\nu$ be two probability measures on a measurable space $(\Omega, \mathcal{F})$. The TV distance between $\mu$ and $\nu$ is defined as
	\[
	\mathrm{TV}(\mu, \nu) = \sup_{A \in \mathcal{F}} |\mu(A) - \nu(A)|.
	\]
	If $\mu$ and $\nu$ are both absolutely continuous w.r.t. a common measure $\lambda$ (e.g., Lebesgue measure), with densities $f = \frac{d\mu}{d\lambda}$ and $g = \frac{d\nu}{d\lambda}$ respectively, then the TV distance can also be written as
	\[
	\mathrm{TV}(\mu, \nu) = \frac{1}{2} \int_\Omega |f(x) - g(x)| \, d\lambda(x).
	\]
	The TV distance takes values in $[0,1]$, where $0$ indicates the distributions are identical, and $1$ indicates they are mutually singular.

	For P-MALA, the proposed update is
	\begin{equation}\label{eq:MALA}
		\tag{P-MALA}
		X^* = \prox_{\tau (F + G \circ K)}(X_n) + \sqrt{2\tau} Z_{n+1},
	\end{equation}
	accepted as $X_{n+1}$ with probability $\min\left\{1, \frac{\pi(X^*) q(X_n | X^*)}{\pi(X_n) q(X^* | X_n)}\right\}$, where 
	$$q(x | y) = \mathcal{N}\Big( \prox_{\tau (F + G \circ K)}(y), \; 2\tau I_d\Big),$$
	and $\mathcal{N}(\mu, \Sigma)$ denotes the multivariate Gaussian density with mean $\mu$ and covariance $\Sigma$; otherwise, $X_{n+1} = X_n$. For MYULA, the update is
	\begin{equation}
		\tag{MYULA}
		X_{n+1} = \left(1 - \frac{\tau}{\theta}\right) X_n - \tau \nabla F(X_n) + \frac{\tau}{\theta} \prox_{\theta G \circ K}(X_n) + \sqrt{2\tau} Z_{n+1},
	\end{equation}
	where $\theta > 0$ is the Moreau-Yosida parameter and the step size satisfies $\tau \leq \frac{\theta}{\theta L_{\nabla F} + 1}$. MYULA approximates the non-smooth term $G \circ K$ with its differentiable Moreau envelope
	\begin{equation}
		(G \circ K)^\theta(x) = \min_{z} \frac{1}{2\theta} \|z - x\|^2 + G(Kz),
	\end{equation}
	which converges to $G \circ K$ as $\theta \to 0$. Applying the ULA to the potential $U^\theta = F + (G \circ K)^\theta$ yields the MYULA iteration, approximating the target density in total variation as $\theta \to 0$. We follow  \cite{habring2024subgradient} in setting $\theta = 0.01$, which achieved favorable convergence.
	
	In Figure 3 of \cite{habring2024subgradient}, the authors evaluated the TV distance for Prox-sub using both the iterates $(\nu_n^N)_{n\in \mathbb{N}}$ with $N = 0$ and the measure $\mu Q_n^\tau S\tau^G$, with the latter demonstrating superior performance across all step sizes $\tau$. Here, we compare only the optimal results of Prox-sub. Figure \ref{fig:Comparison_TV} illustrates the comparison of MASLA with Prox-sub \citep{habring2024subgradient}, P-MALA \citep{pereyra2016proximal}, and MYULA \citep{durmus2022proximal} for $\TV-L^2$ sampling in TV distance. Experiments were conducted with tested step sizes $\tau \in \{10^{-3}, 10^{-4}, 10^{-5}\}$. As shown in Figure \ref{fig:Comparison_W2}, all algorithms exhibit exponential ergodicity to the target density for all tested $\tau$ values.

\subsubsection{Numerical calculation of Wasserstein and TV errors}
\label{sec:numerical_distance_calculation}
For the two-dimensional convex benchmark, the reference-density construction and error evaluation follow \citet[Section~6.1.1]{habring2024subgradient} and their evaluation notebook \citep{habring2024subgradientsource}. The paper describes numerical integration of the target and comparison with sampled distributions; the implementation details below specify the discrete quantities actually evaluated. These calculations concern distributions across independent chains at a fixed iteration, rather than the time-averaged samples of the one-dimensional experiment in Section~\ref{sec:nonapplicable}.

\paragraph{Discretized reference distribution.}
The notebook uses the square $\Omega=[-5,5]^2$, partitioned into $100\times100$ rectangular bins $B_j$ of side length $0.1$. For $U=F+G\circ K$, it numerically integrates
\[
Z_\Omega=\int_\Omega e^{-U(x)}\,dx,
\qquad
b_j=Z_\Omega^{-1}\int_{B_j}e^{-U(x)}\,dx,
\]
and renormalizes the computed bin masses to sum to one. Thus the reference is obtained by quadrature, without drawing target samples. Write $c_j$ for the center of bin $B_j$ and $J=10{,}000$ for the number of bins. For transport computations the target is represented by $\pi_h=\sum_{j=1}^{J}b_j\delta_{c_j}$; for density comparisons it is represented by the piecewise constant density $p_h(x)=b_j/|B_j|$ on $B_j$.

\paragraph{Wasserstein error.}
Let $x_k^{(1)},\ldots,x_k^{(n)}$ be the states of independent chains at iteration $k$, with empirical measure $\widehat\mu_k=n^{-1}\sum_{r=1}^{n}\delta_{x_k^{(r)}}$. The notebook evaluates the squared transport cost
\begin{equation}
\label{eq:discrete_wasserstein_error}
\widehat W_{2,k}^{\,2}
=\min_{P\geq0}\sum_{r=1}^{n}\sum_{j=1}^{J}
P_{rj}\|x_k^{(r)}-c_j\|^2,
\qquad
\sum_jP_{rj}=\frac1n,\quad \sum_rP_{rj}=b_j.
\end{equation}
This is computed using the Python Optimal Transport package \citep{flamary2021pot}: \texttt{ot.dist} constructs the squared Euclidean cost matrix and \texttt{ot.emd2} returns its minimum transport cost. In particular, despite its name, the latter function returns $W_2^2$ for this cost, so a square root is needed to report $W_2$. The authors' Wasserstein notebook uses the first $5{,}000$ stored chains and evaluates every $500$ iterations; its convergence figure plots $\log_{10}\widehat W_{2,k}^{\,2}$. The sample locations themselves form the source support in \eqref{eq:discrete_wasserstein_error}; only the target is replaced by bin centers.

\paragraph{TV error and normalization convention.}
For TV, both distributions are compared on the same bins. Let $N_{k,j}$ count the sampled states in $B_j$ and set $a_{k,j}=N_{k,j}/\sum_\ell N_{k,\ell}$. The density histogram returned by \texttt{numpy.histogram2d} with \texttt{density=True} is multiplied by the bin areas to obtain these probability masses. Under the convention $\mathrm{TV}(\mu,\nu)=\sup_A|\mu(A)-\nu(A)|$ used above, the histogram-based estimate is
\[
\widehat{\mathrm{TV}}_k
=\frac12\sum_{j=1}^{J}|a_{k,j}-b_j|
=\frac12\int_\Omega|\widehat p_{k,h}(x)-p_h(x)|\,dx,
\]
where $\widehat p_{k,h}(x)=a_{k,j}/|B_j|$ on $B_j$. However, the authors' notebook computes $\sum_j|a_{k,j}-b_j|$ without the factor $1/2$, using the total variation norm of the signed difference, with values in $[0,2]$. The two-dimensional comparison code used here follows that convention for Figure~\ref{fig:Comparison_TV}. Its plotted TV errors are therefore twice the distance defined above; converting the logarithmic curves to the $[0,1]$ convention subtracts $\log_{10}2$ from every ordinate and does not change their relative ordering or slopes.

	\subsection{MASLA Performance in Scenarios Beyond Competing Methods}
	\label{sec:nonapplicable}
	
	So far we have shown that MASLA matches the superior performance when those leading algorithms are applicable. Now, we are going to reveal MASLA’s strengths in cases where these methods are not applicable. 	We provide a simple example of a function, $ f(x) = |x^2 - 1| $, defined for all $ x \in \mathbb{R} $, which is locally Lipschitz continuous, lacks a single-valued proximal operator at some points, and has an explicit Clarke subgradient.
	
	\subsubsection{Locally Lipschitz Property}
	
	A function $f$ is locally Lipschitz continuous at $x_0 \in \mathbb{R}$, if there exists a neighborhood of $x_0$ and a constant $L > 0$ such that $|f(x) - f(y)| \leq L |x - y|$ for all $x, y$ in that neighborhood. We verify this for $f(x) = |x^2 - 1|$ as below:
	\begin{itemize}
		\item \textit{Differentiable points}: The function is differentiable except at $x = \pm 1$, where $x^2 = 1$. For $x^2 \neq 1$:
		\begin{itemize}
			\item If $x^2 > 1$ ($x > 1$ or $x < -1$), then $f(x) = x^2 - 1$, so $f'(x) = 2x$.
			\item If $x^2 < 1$ ($-1 < x < 1$), then $f(x) = 1 - x^2$, so $f'(x) = -2x$.
		\end{itemize}
		The derivative $f'(x) = \pm 2x$ is continuous in regions where $f$ is differentiable. In a compact interval $[a, b]$, $|f'(x)| = |2x| \leq 2 \max(|a|, |b|)$. By the mean value theorem, for $x, y$ in a differentiable region, there exists $\xi$ such that:
		\[
		|f(x) - f(y)| = |f'(\xi)| |x - y| \leq 2 \max(|a|, |b|) |x - y|.
		\]
		\item \textit{Non-differentiable points}: At $x = 1$, for small $h$, $$ f(1 + h) = |(1 + h)^2 - 1| = |2h + h^2| \approx 2|h| ,$$ and $f(1) = 0$. Thus:
		\[
		|f(1 + h) - f(1)| \approx 2|h|,
		\]
		indicating local Lipschitz continuity with constant approximately 2. Similarly at $x = -1$. 
	\end{itemize}		
		Hence, $f$ is locally Lipschitz on $\mathbb{R}$, with the Lipschitz constant depending on the neighborhood size due to the linear growth of $f'(x)$.

	\subsubsection{Proximal Operator Not Well-Defined Everywhere}
	
	The proximal operator of $f$ at $y \in \mathbb{R}$ with parameter $\lambda > 0$ is given by
	\[
	\text{prox}_{\lambda f}(y) = \arg\min_{x \in \mathbb{R}} \left\{ f(x) + \frac{1}{2\lambda} (x - y)^2 \right\}.
	\]
	For $\text{prox}_{\lambda f}(y)$ to be well-defined, the minimizer must be unique. Set $y = 0$, $\lambda = 1$ and define
	\[
	\phi(x) := |x^2 - 1| + \frac{1}{2} x^2.
	\]
	We present our discussion below:
	\begin{itemize}
		\item For $x^2 < 1$, we have $|x^2 - 1| = 1 - x^2$, and then
		\[
		\phi(x) = 1 - x^2 + \frac{1}{2} x^2 = 1 - \frac{1}{2} x^2.
		\]
		\item For $x^2 > 1$, we have $|x^2 - 1| = x^2 - 1$, and then
		\[
		\phi(x) = x^2 - 1 + \frac{1}{2} x^2 = \frac{3}{2} x^2 - 1.
		\]
		\item For $x = \pm 1$, 
		\[
		\phi(1) = \frac{1}{2}\quad\text{and}\quad \phi(-1) = \frac{1}{2}.
		\]
	\end{itemize}
	The function $\phi(x)$ is maximized at $x = 0$ ($\phi(0) = 1$) for $|x| < 1$, decreases to $\phi(\pm 1) = \frac{1}{2}$, and increases for $|x| > 1$. The global minima occur at $x = \pm 1$, yielding:
	\[
	\text{prox}_{f}(0) = \{ -1, 1 \}.
	\]
	Since the minimizer is not unique, the proximal operator is not single-valued at $y = 0$, and thus not well-defined everywhere as a single-valued mapping, due to the non-convexity of $f$.
	
	\subsubsection{Explicit Clarke Subgradient}
	
	By the Clarke subgradient definition provided in Definition \ref{def:Clarke}, since $f$ is differentiable except at $x = \pm 1$, 
	\begin{itemize}
		\item For $x > 1$ or $x < -1$, we have$f'(x) = 2x$, and then $\partial f(x) = \{ 2x \}$.
		\item For $-1 < x < 1$, we have $f'(x) = -2x$, and then $\partial f(x) = \{ -2x \}$.
	\item For $x = 1$, the right limit as $x_i \to 1^+$ yields $f'(x_i) = 2x_i \to 2$, while the left limit as $x_i \to 1^-$ gives $f'(x_i) = -2x_i \to -2$. Consequently, the Clarke subdifferential is $\partial f(1) = \text{conv} \{ -2, 2 \} = [-2, 2]$.
\item For $x = -1$, the right limit as $x_i \to -1^+$ yields $f'(x_i) = -2x_i \to 2$, while the left limit as $x_i \to -1^-$ gives $f'(x_i) = 2x_i \to -2$. Consequently, the Clarke subdifferential is $\partial f(-1) = \text{conv} \{ -2, 2 \} = [-2, 2]$.
	\end{itemize}
	Thus, the Clarke subgradient is given by
	\[
	\partial f(x) =
	\begin{cases} 
		\{ 2x \} & \text{if } x < -1 \text{ or } x > 1, \\
		[-2, 2] & \text{if } x = \pm 1, \\
		\{ -2x \} & \text{if } -1 < x < 1.
	\end{cases}
	\]
	
	\subsubsection{The Strength of the MH Correction}
	The multi-valued nature of the proximal operator hinders  the application of proximal-based methods, however we can still compare with the algorithmic version without the MH correction to numerically visualize the strength of the MH correction. \cite{difonzo2022stochastic} focused on the theoretical analysis of stochastic Langevin differential inclusions of continuous-time while they ran unadjusted Langevin dynamics with the Euler-Maruyama discretization. Here, we refer it as Unadjusted Subdifferential Langevin Algorithm (USLA), which simply draws a subgradient from the set-valued field at each iteration:
	\begin{equation}
		\label{eqn:USLA}
		\breve{X}_{n+1} = \breve{X}_n - \gamma \beta(\breve{X}_n) + \sqrt{2 \gamma}\, Z_{n+1}, \qquad \beta(\breve{X}_n) \in D_U(\breve{X}_n),
	\end{equation}
	We initialized both USLA and MASLA chains at $x = 0$, ran them for 100,000 iterations with a step size of $\tau = 0.1$, and discarded the first 20\% of iterations as burn-in. Figure \ref{fig:usla_vs_masla} compares the performance of USLA and MASLA chains against the true density. MASLA closely aligns with the true density, whereas USLA deviates significantly. The total variation (TV) distance errors are 0.116761 for USLA and 0.014363 for MASLA, while the Wasserstein-2 distance errors are 0.092183 for USLA and 0.008199 for MASLA. The computation times are 0.9976 seconds for USLA and 1.5361 seconds for MASLA, measured on a personal laptop equipped with an Apple M1 Pro chip, featuring 10 cores and 16 GB of memory.

\subsection{A high-dimensional nonsmooth nonconvex quartic target}
\label{sec:high_dimensional_quartic}

The two-dimensional convex benchmark in Section~\ref{sec:Comparative} and the one-dimensional nonconvex example in Section~\ref{sec:nonapplicable} permit direct visualization of the target distribution. To complement these examples, we compare MASLA and USLA on a $100$-dimensional target with a superlinearly growing drift. This experiment examines numerical escape of the unadjusted discretization and the effect of the MH correction in a higher-dimensional setting.

\paragraph{Target and subgradient selection.}
For $d=100$, let
\begin{equation}
\label{eq:high_dimensional_quartic_target}
\pi_d(x)=Z_d^{-1}e^{-U_d(x)},\qquad
U_d(x)=\sum_{i=1}^d u(x_i),\qquad
u(t)=\frac{t^4}{4}-t^2+|t|.
\end{equation}
The potential is locally Lipschitz and semialgebraic, and is nonsmooth on the coordinate hyperplanes. It is nonconvex: at points with all coordinates nonzero, its Hessian is $\operatorname{diag}(3x_i^2-2)$, which has a negative eigenvalue whenever some $0<|x_i|<\sqrt{2/3}$. The quartic terms ensure normalizability; indeed, $u(t)\geq t^4/8-2$. Each coordinate potential has strict local minima at $0$ and $\pm1$, giving $3^d$ strict local modes of the product density. We use the Clarke conservative field and its Borel selection
\begin{equation}
\label{eq:high_dimensional_quartic_selection}
\partial U_d(x)=\prod_{i=1}^d\bigl(x_i^3-2x_i+\partial|x_i|\bigr),
\qquad
\beta_i(x)=x_i^3-2x_i+\operatorname{sign}(x_i),
\end{equation}
where $\partial|0|=[-1,1]$ and $\operatorname{sign}(0)=0$. Although the target factorizes, both algorithms update the entire $100$-dimensional vector at each iteration.

\paragraph{Simulation protocol.}
For each $\gamma\in\{0.05,0.10,0.20\}$, we simulate $16$ independent replicates of each algorithm, initialized at $X_0=0.1\mathbf{1}_d$, for up to $10{,}000$ updates. Both use proposals $Y=X-\gamma\beta(X)+\sqrt{2\gamma}\,Z$, with $Z\sim N(0,I_d)$. USLA accepts every proposal, as in \eqref{eqn:USLA}. MASLA applies one joint MH decision using the full target and Gaussian proposal densities in \eqref{eq:masla_acceptance}; the log acceptance ratio is evaluated directly for numerical stability. Independent random streams for the algorithms and step sizes are generated from NumPy's \texttt{SeedSequence} with seed $42$.

For each USLA replicate, we record the first threshold-crossing time
\[
T_{\mathrm{esc}}=\inf\{n\geq1:\|X_n\|_\infty\geq10^6\}.
\]
The crossing value is retained, and that replicate is then stopped to avoid overflow; paths are neither clipped nor restarted. Runs without a crossing are observed through iteration $10{,}000$. Thus the reported escape fraction uses all $16$ replicates, including those stopped early. For MASLA, we discard the first $2{,}000$ updates and pool coordinate $1$ from the remaining $8{,}000$ states of each replicate, including repeated states from rejections. The resulting histogram contains $128{,}000$ correlated observations. Its reference curve is the marginal density proportional to $e^{-u(t)}$, normalized numerically by the trapezoidal rule on $4{,}001$ equally spaced points in $[-4,4]$. Acceptance rates are computed over all $10{,}000$ proposals and averaged across replicates.

\paragraph{Results.}
Figure~\ref{fig:high_dimensional_quartic} reports the marginal histograms, individual trajectory magnitudes, and cumulative escape fractions. No USLA replicate crossed the threshold at $\gamma=0.05$ or $0.10$ during the observation period. At $\gamma=0.20$, all $16$ replicates crossed it, with first-crossing times ranging from $10$ to $336$ iterations. The cubic drift explains the potential for large excursions: in each coordinate the noiseless proposal is $(1+2\gamma)t-\gamma t^3-\gamma\operatorname{sign}(t)$, whose magnitude grows cubically for large $|t|$.

MASLA remained below the escape threshold in all runs. Its mean acceptance rates were $58.03\%$, $21.49\%$, and $1.37\%$ at the three respective step sizes. The marginal histograms agree reasonably with the reference at the smaller step sizes, whereas the low acceptance rate and irregular histogram at $\gamma=0.20$ indicate poor exploration over the available run length. Thus the MH correction prevents the observed explosive excursions in this experiment, but a large step size still compromises sampling efficiency. These finite simulations provide evidence of numerical instability of USLA. 

\begin{figure}[t!]
\centering
\includegraphics[width=\linewidth]{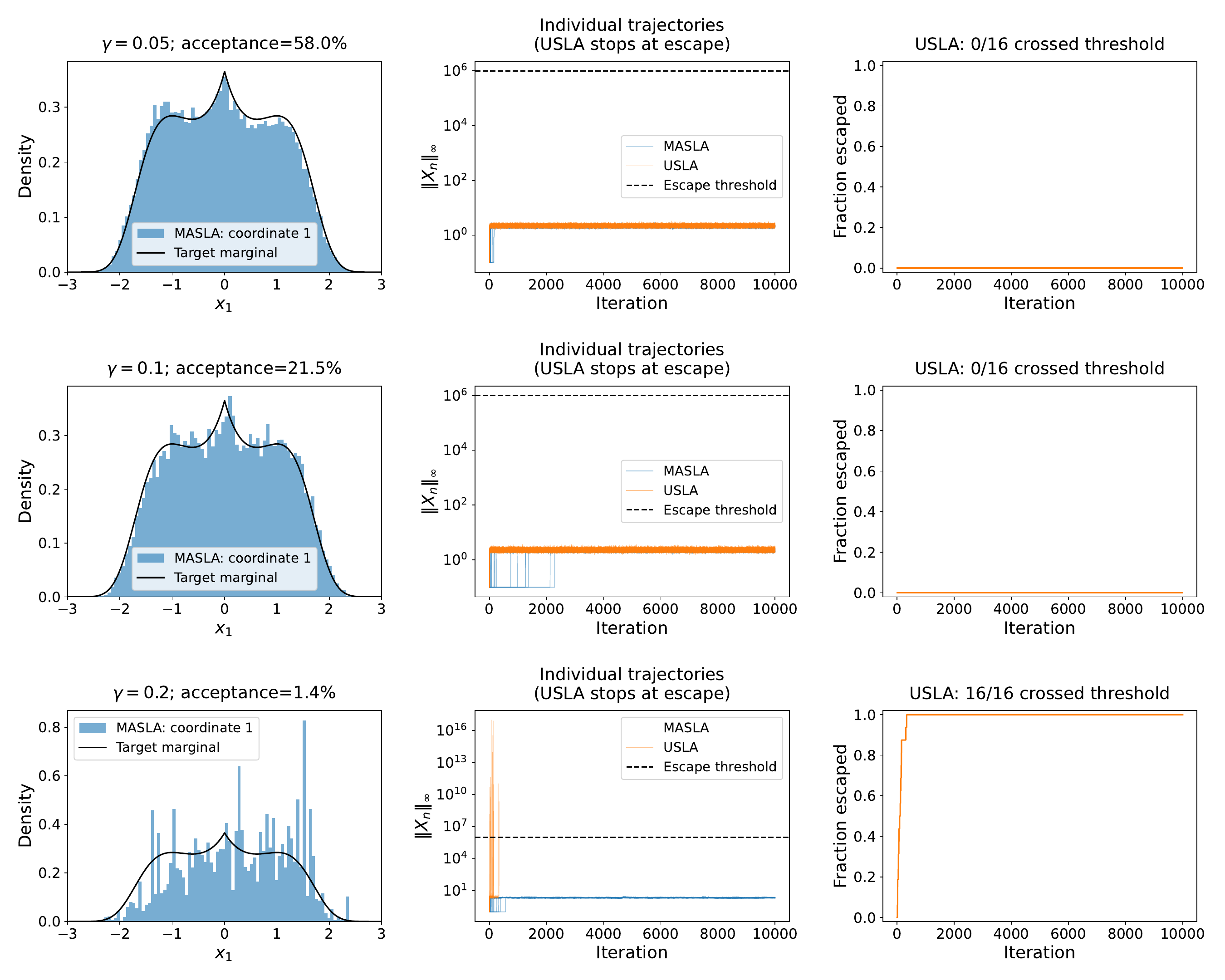}
\caption{MASLA and USLA for the $100$-dimensional nonsmooth nonconvex quartic target \eqref{eq:high_dimensional_quartic_target}. Rows correspond to $\gamma=0.05$, $0.10$, and $0.20$. Left: MASLA coordinate-$1$ histograms after burn-in and the target marginal. Middle: $\|X_n\|_\infty$ for all replicates, on a logarithmic scale; the dashed line is the $10^6$ escape threshold, and USLA trajectories stop after crossing it. Right: the fraction of the $16$ USLA replicates that have crossed the threshold by each iteration. All USLA replicates escape at $\gamma=0.20$, while MASLA has a low acceptance rate at that step size.}
\label{fig:high_dimensional_quartic}
\end{figure}

\end{document}